\documentclass[%
 reprint,
superscriptaddress,
 amsmath,amssymb,
 aps,
]{revtex4-2}

\usepackage{graphicx}
\usepackage{dcolumn}
\usepackage{bm}

\usepackage[utf8]{inputenc}
\usepackage[english]{babel}
\usepackage[T1]{fontenc}
\usepackage{amsmath}
\usepackage{amssymb}
\usepackage[colorlinks=true, linkcolor=blue, citecolor=blue, urlcolor=blue, pdfborder={0 0 0}]{hyperref}
\usepackage{amsmath,amssymb,amsthm,mathtools,epsfig,color,empheq,graphicx,graphics,algorithm,algorithmic,braket,dsfont,accents,multirow, float, balance, accents}
\usepackage{xr-hyper}
\newcommand{\ubar}[1]{\underaccent{\bar}{#1}}

\DeclareMathOperator{\diag}{dg}

\DeclareMathOperator{\trace}{Tr}

\DeclareMathOperator{\real}{Re}
\DeclareMathOperator{\imag}{Im}

\DeclareMathOperator{\Var}{Var}

\newtheorem{assumption}{Assumption}
\newtheorem{proposition}{Proposition}
\newtheorem{lemma}{Lemma}

\newtheorem{theorem}{Theorem}

\newcommand \bzero{\mathbf{0}}
\newcommand \bone{\mathbf{1}}

\newcommand \bb{\mathbf{b}}

\newcommand \be{\mathbf{e}}
\newcommand \bg{\mathbf{g}}

\newcommand \bi{\mathbf{i}}

\newcommand \bp{\mathbf{p}}

\newcommand \bv{\mathbf{v}}

\newcommand \bx{\mathbf{x}}

\newcommand \bz{\mathbf{z}}

\newcommand \bB{\mathbf{B}}

\newcommand \bG{\mathbf{G}}
\newcommand \bH{\mathbf{H}}
\newcommand \bI{\mathbf{I}}
\newcommand \bJ{\mathbf{J}}

\newcommand \bL{\mathbf{L}}
\newcommand \bM{\mathbf{M}}
\newcommand \bN{\mathbf{N}}

\newcommand \bP{\mathbf{P}}

\newcommand \bR{\mathbf{R}}
\newcommand \bS{\mathbf{S}}
\newcommand \bT{\mathbf{T}}
\newcommand \bU{\mathbf{U}}
\newcommand \bV{\mathbf{V}}
\newcommand \bW{\mathbf{W}}
\newcommand \bX{\mathbf{X}}
\newcommand \bY{\mathbf{Y}}

\newcommand \bzeta{\boldsymbol{\zeta}}
\newcommand \btheta{\boldsymbol{\theta}}

\newcommand \blambda{\boldsymbol{\lambda}}

\newcommand \bxi{\boldsymbol{\xi}}

\newcommand \brho{\boldsymbol{\rho}}

\newcommand \bphi{\boldsymbol{\phi}}

\newcommand \bpsi{\boldsymbol{\psi}}

\newcommand \bLambda{\mathbf{\Lambda}}

\newcommand \bSigma{\mathbf{\Sigma}}

\newcommand \htE{\hat{E}}
\newcommand \htF{\hat{F}}

\newcommand \mcC{\mathcal{C}}

\newcommand \mcE{\mathcal{E}}

\newcommand \mcG{\mathcal{G}}

\newcommand \mcL{\mathcal{L}}

\newcommand \mcN{\mathcal{N}}
\newcommand \mcO{\mathcal{O}}

\newcommand \mcV{\mathcal{V}}

\newcommand \mcZ{\mathcal{Z}}

\newcommand \cbM{\check{\mathbf{M}}}

\newcommand \cbLambda{\check{\mathbf{\Lambda}}}

\newcommand \hbg{\hat{\mathbf{g}}}
\newcommand \hbz{\hat{\mathbf{z}}}

\newcommand \hbM{\hat{\mathbf{M}}}

\newcommand \hbLambda{\hat{\mathbf{\Lambda}}}

\newcommand \bbz{\bar{\mathbf{z}}}

\begin{document}
\preprint{APS/123-QED}

\title{Solving Conic Programs over Sparse Graphs using a Variational Quantum Approach:\\
The Case of the {AC} Optimal Power Flow}

\author{Thinh Viet Le}
\email{le272@purdue.edu}
\affiliation{Elmore Family School of Electrical and Computer Engineering,
Purdue University, West Lafayette, Indiana 47906, USA}

\author{Mark M. Wilde}
\email{wilde@cornell.edu}
\affiliation{School of Electrical and Computer Engineering,
Cornell University, Ithaca, New York 14853, USA}

\author{Vassilis Kekatos}
\email{kekatos@purdue.edu}
\affiliation{Elmore Family School of Electrical and Computer Engineering,
Purdue University, West Lafayette, Indiana 47906, USA}

\date{\today}

\begin{abstract}
{Conic programs arising in physics, quantum information, machine learning, and engineering, are oftentimes defined over sparse graphs.} Although such problems can be solved in polynomial time using classical interior-point solvers, the computational complexity scales unfavorably with the graph size. In this context, this work proposes a variational quantum paradigm for solving conic programs, including quadratically constrained quadratic programs (QCQPs) and semidefinite programs (SDPs). We encode primal variables via the state of a parameterized quantum circuit (PQC) and dual variables via the probability mass function associated with a second PQC. The Lagrangian function can thus be expressed as scaled expectations of quantum observables. We pursue approximately stationary points of the Lagrangian by minimizing/maximizing the Lagrangian over the parameters of the first/second PQC. This is accomplished in a hybrid fashion: Gradients of the Lagrangian are estimated using the two PQCs, while PQC parameters are updated classically using a primal-dual method. We propose permuting the primal variables so that related observables are expressed in a banded form, enabling efficient measurement. {We provide a complexity analysis that is useful to determine which problem types may enjoy quantum advantage.} The proposed framework is applied to the AC OPF problem, a large-scale optimization problem central to the operation of electric power systems. Numerical tests on the IEEE 57-node power system using Pennylane's simulator corroborate that the proposed doubly variational quantum framework can find high-quality OPF solutions. {While this OPF demonstration does not yield a quantum speedup, the results serve as a proof-of-concept and highlight challenges toward practical quantum advantage.} Although showcased for the OPF, this framework has a broader scope, including conic programs with many variables and constraints, problems defined over sparse graphs, and training quantum machine learning models to satisfy constraints.
\end{abstract}

\maketitle


\section{Introduction}
Conic programs defined over sparse graphs are ubiquitous in physics, quantum information, machine learning, and engineering. Many of these problems naturally take the form of semidefinite programs (SDPs) or nonconvex quadratically-constrained quadratic programs (QCQPs) that can be lifted to SDPs. In many of these formulations, the number of constraints scales with the graph size, and thus, although each constraint involves a small subset of variables, the Lagrangian function contains numerous terms. Representative examples include finding the ground-state energy of an Ising spin system~\cite{Barahona}, computing graph-theoretic quantities such as the Lov\'{a}sz theta number~\cite{lovasz}, verifying deep neural network properties~\cite{raghunathan}, and networked optimization in engineering, such as the optimal power flow problem~\cite{bai2008, lavaei2011}. Although classical algorithms for solving SDPs, such as interior-point methods~\cite{BoVa04}, are guaranteed to find a globally optimal solution, their computational cost may scale unfavorably with the problem dimension. 

Variational quantum algorithms (VQAs) have emerged as a versatile and promising paradigm for leveraging noisy intermediate-scale quantum (NISQ) processors to solve complex optimization problems. By employing parameterized quantum circuits (PQCs) optimized via classical routines, these hybrid quantum-classical algorithms are designed to comply with the constraints of current quantum hardware. Nonetheless, their application to constrained optimization remains underdeveloped. Existing methodologies are mostly limited to problems with binary variables and/or linear constraints, leaving a gap for the more general yet powerful class of conic programs, such as QCQPs and SDPs. Recently, Reference~\cite{patel2024} used two PQCs to encode both primal and dual variables (Lagrange multipliers) for solving SDPs with many constraints. Nevertheless, it assumes efficient measurement of observables and relies on a dual-decomposition method that entails a nested variational optimization. 



\subsection{Contributions}
This work proposes a novel VQA framework for solving large-scale conic programs with numerous constraints defined over sparsely connected graphs.
The contributions are on three fronts:

\emph{c1) Optimizing a doubly parameterized Lagrangian
function.} Building on the modeling proposed in~\cite{patel2024}, where primal variables are captured by the state of a primal PQC and the dual variables are captured through the probability mass function (PMF) associated with a dual PQC, we develop a VQA seeking an approximately stationary point of the associated Lagrangian function. We aim to minimize the Lagrangian function with respect to the parameters of the primal PQC and maximize it with respect to the parameters of the dual PQC. This is accomplished in a hybrid fashion. First, the two PQCs are used to measure gradients of the Lagrangian with respect to the PQC parameters. A classical computer subsequently uses gradient information to iteratively update PQC parameters using a primal-dual method. In contrast to the dual decomposition~\cite{patel2024,ICASSP23}, which requires solving an inner variational optimization problem to optimality at each iteration, our method performs a single gradient descent step on the primal variables and a single gradient ascent step on the dual variables per iteration. The proposed primal-dual updates apply directly to a broader class of optimization problems that can be solved via VQAs. Moreover, the proposed method enables the training of quantum machine learning (QML) models that satisfy constraints. Such functionality is essential for endowing QML models with safety and stability features. 

\emph{c2) Efficient measurement of graph-induced observables.} Although the concerned observables are defined over sparse Hermitian matrices, standard LCU (linear combination of unitaries)-based decomposition methods do not yield qubit-efficient measurement protocols. To address this challenge, we utilize the graph structure of the problem instance and the extended Bell measurements (XBM)~\cite{Kondo2022}, to design a qubit-efficient protocol for measuring gradients of the Lagrangian function associated with the QCQP/SDP. Specifically, the XBM method partitions the entries of an $N\times N$ Hermitian measurement matrix into groups, allowing entries within each group to be measured simultaneously using at most $\mcO(\log N)$ additional quantum gates. If a matrix can be decomposed into $\tilde{\mcO}(\log N)$ such groups, the XBM method is qubit-efficient. Hence, we propose to permute the graph's nodes so that the related observables are banded. Finding the node permutation that yields the optimal number of groups is an NP-hard problem~\cite{feige2000}. Instead, we
adopt the reverse Cuthill-McKee (RCM) algorithm~\cite{rcm}, a greedy heuristic algorithm of linear complexity $\mcO(N)$. This is effective because banded matrices are known to feature a small number of XBM groups~\cite{Kondo2022}. The devised workflow can be instrumental in effectively measuring quantum observables associated with sparse connectivity graphs arising in applications over natural or engineered networked systems, including the representative examples discussed earlier. {In Section~\ref{sec:analysis}, we provide a sample complexity analysis that identifies the conditions under which the proposed doubly variational framework can achieve quantum advantage.}

\emph{c3) Application to the Optimal Power Flow (OPF) problem.} We illustrate how the proposed framework applies to the OPF problem. OPF can be posed as an optimization problem over thousands of variables and constraints. If power flows need to be modeled precisely through the alternating-current (AC) power flow equations, the OPF can be formulated as a QCQP. Although this QCQP is not convex, the OPF can be relaxed to a convex SDP and solved to global optimality under certain conditions~\cite{bai2008,lavaei2011}. Nonetheless, the complexity of modern-day electric power grids raises scalability and optimality
challenges. In this context, VQA offers a promising alternative, provided that the required quantum measurements can be performed efficiently. Two key points for adopting the XBM method to the OPF are: \emph{i)} The sparsity pattern (the positions of nonzero entries) of OPF observables is determined by the graph defined by the nodes and transmission lines of the power system, which is sparse; and \emph{ii)} According to numerical tests, nodes in a power system can be permuted via the RCM algorithm so that OPF observables can be partitioned in $\mcO((\log N)^3)$ groups. {Although the sample complexity of the proposed model for OPF does not scale well with $N$, the framework may still be meaningful in a quantum machine learning setting, where the PQCs are trained offline and subsequently used to generate efficient solutions during operation. Moreover, the XBM quantum measurement protocol can be instrumental in quantum implementations of power grid operations and, more broadly, in computational tasks over sparsely connected graphs.}

\subsection{Literature review}

\emph{Handling constrained optimization problems using VQAs.} 
A well-studied exemplar of VQAs is the variational quantum eigensolver (VQE), which aims to find the smallest eigenvalue of a quantum observable by employing a PQC~\cite{peruzzo2014}. The quantum approximate optimization algorithm (QAOA) is a variant of VQE with a problem-dependent PQC and a diagonal observable targeting combinatorial problems~\cite{farhi2014}. While VQE and QAOA were originally devised for unconstrained problems, extensions to optimization problems with constraints have been proposed in~\cite{Lucas, ICASSP23,VQEC, patel2024,westerheim2024,Hadfield19,bharti2021,VQA4QCQP1}. References~\cite{Lucas,westerheim2024} convert constrained problems to unconstrained ones by penalizing constraint violations on the objective function. However, selecting a suitable penalty parameter is a nontrivial task. Unreasonably large values can cause ill-conditioning, while unreasonably small values can yield infeasible solutions. In~\cite{Hadfield19}, the QAOA's ansatz is adapted to ensure that the output state stays within the feasible subspace; yet, this strategy applies only to binary problems with linear constraints. References~\cite{bharti2021} propose a hybrid method for solving SDPs. The idea is to approximate the matrix variable of the original SDP by a semidefinite matrix of smaller dimension. The resultant SDP of reduced dimension is solved classically using standard interior point-based methods. Quantum measurements are used to translate the cost and constraint functions of the original to the compressed SDP formulation. Nonetheless, the solution of the compressed SDP may be infeasible or suboptimal for the original SDP.

An alternative approach to cope with constrained problems is via Lagrangian duality. Reference~\cite{ICASSP23} seeks a saddle point of the Lagrangian function with respect to the PQC parameters and the dual variables. To tackle problems with numerous constraints, reference~\cite{patel2024} encodes dual variables through a parameterized quantum state, in addition to the primal variables.  Both~\cite{ICASSP23} and \cite{patel2024} solve the associated dual problem using the iterative method of dual decomposition. Nevertheless, each iteration of dual decomposition involves solving a VQE-type problem to optimality, which is computationally overwhelming. In~\cite{VQEC}, this drawback is circumvented by using the primal-dual update method, which only requires estimating gradients of the Lagrangian function. Nonetheless, reference~\cite{VQEC} applies only to binary and linear programs. In this work, we extend~\cite{VQEC} and \cite{patel2024} to deal with a doubly variational VQA formulation of QCQP/SDPs, which we showcase via the OPF problem.

\emph{Measuring quantum observables.} Measuring quantum observables is key to the effective implementation of VQAs. A prominent measurement method decomposes a Hermitian matrix into a linear combination of unitaries (LCU), and measures the resultant unitary observables via the destructive swap test~\cite{schuld2021,swap_test}. This method is practically relevant only if the resultant unitaries can be implemented efficiently using quantum gates. This is the case, for example, if the unitaries can be expressed as Pauli strings (tensor products of Pauli matrices), and the number of Pauli strings scales polynomially with the number of qubits, $n$; see~\cite{mcclean2016,crawford2021,nielsen00}. Unfortunately, a general matrix may involve $\mcO(4^n)$ Pauli strings~\cite{koska2024}. The number of Pauli strings grows exponentially with $n$, even for some structured sparse matrices, such as $k$-banded matrices that involve $\mcO(kn2^n)$ Pauli strings~\cite{koska2024}. One possible strategy to reduce the number of measurements is to group Pauli strings into sets of terms that pairwise commute. If two Hermitian matrices commute, they can be jointly diagonalized, and thus, measured simultaneously~\cite[pg.~43]{shankar2012}. To this end, various grouping methods have been proposed, including qubit-wise commutativity grouping~\cite{verteletskyi2020,yen2020}, general commutativity grouping~\cite{gokhale2020}, and unitary partitioning~\cite{izmaylov2019,zhao2020}. Nevertheless, partitioning Pauli strings to achieve the minimum number of groups is an NP-hard problem~\cite{gokhale2020}. 

Classical shadowing is an alternative technique for measuring observables~\cite{huang2020predicting}. It is particularly effective when the measurement matrix can be expressed as local Pauli strings, i.e., Pauli strings that act non-trivially on a small number of qubits. For observables that do not admit this form, however, the number of measurements scales linearly with the squared Frobenius norm of the matrix. Unfortunately, for connectivity matrices of sparse graphs, this quantity generally grows linearly with the network size.

Another recent method for measuring observables involves extended Bell measurements~\cite{Kondo2022}. This method decomposes the measurement matrix into matrices whose associated observables are simultaneously measurable using qubit-efficient rotation unitaries. Nevertheless, the number of induced observables may reach $2^n$ for a general sparse measurement matrix~\cite{Kondo2022}. Our key idea is to permute nodes of the underlying graph once and beforehand, so that the associated quantum observables can be measured efficiently. 

\emph{Quantum computing for the OPF.} A previous quantum computing effort to the OPF problem has focused on the classical primal-dual interior point method, wherein the Harrow–Hassidim–Lloyd (HHL) algorithm serves as a linear system solver in each Newton step~\cite{amani2024quantum}. The idea of utilizing the HHL algorithm as a linear system solver has also been explored for solving the nonlinear AC power flow equations, which constitute a subset of constraints of the AC OPF problem~\cite{Spyros1,feng2021quantum,liu2024quantum}. Nonetheless, the computational complexity of the HHL algorithm does not scale well for linear systems involving the Jacobian matrix of the AC power flow equations, as established in~\cite{demystifying}. Reference~\cite{hu2024} integrates a PQC as an intermediate layer of a classical neural network, trained in a supervised fashion to predict OPF solutions. However, it is not clear
which measurement outputs are encoded into the subsequent classical hidden layer, nor how the parameter-shift rule is applied under such a design.

\subsection{Notation}
Column vectors are denoted by boldface lower-case letters. Matrices are denoted by boldface upper-case letters. The symbol $^\top$ stands for transposition, and $^\dag$ for conjugate transposition. The operator $\diag(\bx)$ returns a diagonal matrix with vector $\bx$ on its main diagonal, while $\diag(\bX)$ returns a vector with the diagonal entries of matrix $\bX$. The set of real-valued symmetric matrices of dimension $M$ is denoted by $\mathbb{S}^M$. Logarithmic functions are of the binary base. The imaginary unit is denoted by $\iota\coloneqq \sqrt{-1}$. The notation $m=1:M$ means that index $m$ takes values from $1$ to $M$. Symbol $\left\|\bH\right\|$ denotes the spectral norm of a matrix $\bH$. {We use $\mcO(\cdot)$ to denote standard asymptotic complexity order, while $\tilde{\mcO}(\cdot)$ suppresses polylogarithmic factors in the problem dimensions.}

\section{Primal-dual iterations for QCQP}\label{sec:classicalpd}
We focus on optimization problems defined over sparse graphs, often arising from physical models as nonconvex QCQPs or SDPs. The latter may appear either as native formulations or as convex relaxations of the QCQPs. Prototypical examples of such problems include:
\begin{itemize}
\item \textbf{Ising ground-state energy:} Minimizing a quadratic Hamiltonian $H(\mathbf{s}) = \sum_{i<j} J_{ij}s_i s_j + \sum_i h_i s_i$ over spin configurations $\mathbf{s} \in \{\pm1\}^N$ can be expressed as a nonconvex QCQP. The coupling matrix $\bJ$ is typically sparse, reflecting local interactions in physical lattices or graphical models. The resulting QCQP admits an SDP relaxation obtained by lifting the quadratic $\pm1$ constraints $s_i^2 = 1$; see~\cite{Barahona}.

\item \textbf{Lov\'asz theta number:} 
For a graph $\mcG=(\mcV,\mcE)$ with $N=|\mcV|$ vertices, the Lov\'asz theta number can be computed as the optimal value of an SDP over an $N\times N$ matrix variable:
\begin{align}
    \vartheta(\mcG):=\max_{\bX \succeq 0} ~&~\bone^\top\bX\bone \nonumber\\
    \textrm{subject to (s.to)}~&~\trace(\bX)=1, \quad  \nonumber\\ 
    ~&~ \bX_{i,j} = 0 \ \forall (i,j)\in\mcE,\nonumber
\end{align}
where $\bone$ is the all-ones vector. Each edge $(i,j)\in\mcE$ yields one constraint, so the constraint matrices inherit the sparsity pattern of $\mcG$. While commonly presented as an SDP, the Lov\'asz theta number admits an equivalent QCQP formulation~\cite{lovasz}.



\item \textbf{Neural network verification:} Certifying the robustness of neural networks (NN) against adversarial perturbations entails bounding the worst‑case margin between classes over a permissible input set. In networks using the rectified linear unit (ReLU) as activation function, the input-output relationship of the ReLU $z = \max(0, x)$ admits an exact reformulation including a quadratic constraint, namely $z \ge 0$, $z \ge x$, and $z(z-x)=0$~\cite{raghunathan}. Based on this reformulation, the NN verification task can be cast as a nonconvex QCQP over the high-dimensional vector of neuron activations. The constraint matrices are highly sparse and structured, reflecting the layered connectivity of the underlying computational graph~\cite{raghunathan}.

\item \textbf{Optimal power flow:} The AC OPF problem minimizes the electricity generation cost to meet load demand while satisfying grid engineering constraints. The problem can be posed as a nonconvex QCQP over nodal voltages, where the quadratic constraints represent Kirchhoff's laws and operational limits. The sparsity pattern is governed by the power grid topology through transmission lines, as each grid node is directly connected to only a few other nodes~\cite{bai2008}.
\end{itemize}

To streamline the presentation, we focus on the QCQP formulation and discuss its extension to the SDP setting in Section~\ref{sec:parameterizedopf}. We consider the QCQP:
\begin{align}\tag{QCQP}
    P^* \coloneqq \min_{\bx\in \mathbb{C}^{N}}~&~\bx^\dag \bM_0 \bx\label{eq:qc}\\
    \textrm{s.to}~&~\bx^\dag \bM_m \bx \leq b_m, \quad \text{for}~~m=1:M \nonumber
\end{align}
where $\{\bM_m \in \mathbb{C}^{N\times N}\}_{m=0}^M$ are sparse Hermitian matrices determined by the underlying problem graph, and the number of constraints $M$ typically scales with the network size $N$.

Although \eqref{eq:qc} is generally nonconvex,  it can be relaxed to an SDP, which can be solved by standard interior-point solvers. In some application domains, this convex relaxation is successful in the sense that the optimal SDP variable is rank-one and a globally optimal $\bx$ can be readily recovered from the SDP solution. Nonetheless, the
memory and time requirements imposed by SDP solvers
may be challenging to meet as $N$ increases. This motivates us to explore a VQA for QCQP/SDPs. 

Before doing so, let us discuss how \eqref{eq:qc} could be solved on a classical computer. Among the various approaches, we present the primal-dual (PD) decomposition as a representative first-order optimization algorithm. The PD method relies on Lagrangian duality~\cite{BoVa04}. Each constraint indexed by $m$ in~\eqref{eq:qc} is associated with a non-negative dual variable $\lambda_m$. Let vector $\blambda \in \mathbb{R}^{M}_+$ collect all dual variables. The related \emph{Lagrangian function} is a function of $\bx$ and $\blambda$:
\begin{equation} \label{eq:Lagrangian0}
\mcL(\bx;\blambda)=\bx^\dag \bM_0 \bx+\sum_{m=1}^M \lambda_m (\bx^\dag \bM_m \bx-b_m).
\end{equation}
The corresponding \emph{dual problem} is posed as
\begin{equation}
\label{eq:qcdual}
D^*\coloneqq \max_{\blambda\succeq \bzero}\min_{\bx\in \mathbb{C}^N}~\mcL(\bx;\blambda).
\end{equation}
A \emph{primal/dual solution} $(\bx^*, \blambda^*) \in \mathbb{C}^N \times \mathbb{R}^M_+$ of \eqref{eq:qcdual} satisfies 
\[\mcL(\bx^*;\blambda) \leq \mcL(\bx^*;\blambda^*) \leq \mcL(\bx;\blambda^*)\quad~\]
for all $(\bx, \blambda)\in \mathbb{C}^N \times \mathbb{R}^M_+$. Such a pair is called a \emph{saddle point} of the Lagrangian function. In essence, $\bx^*$ is a minimizer of $\mcL(\bx;\blambda^*)$ and $\blambda^*$ is a maximizer of $\mcL(\bx^*;\blambda)$. For problems featuring strong duality (in which case $P^*=D^*$), a point is an \emph{optimal primal/dual point} if and only if it is a saddle point of the Lagrangian function~\cite[p.~239]{BoVa04}. Furthermore, at a primal/dual solution $(\bx^*;\blambda^*)$, the value of the Lagrangian satisfies $\mcL(\bx^*;\blambda^*)=P^*=D^*$.

The previous discussion relies on strong duality, which does not generally hold for the nonconvex problem~\eqref{eq:qc}. Problem \eqref{eq:qc} is guaranteed to feature strong duality if its semidefinite relaxation is exact and 
Slater’s condition is satisfied for that relaxation~\cite{kim2025equivalent}. Under these conditions, a 
saddle point of $\mcL$ exists and coincides with a global optimum. Even with strong duality in place, finding a saddle point using an optimization algorithm imposes stringent conditions on $\mcL(\bx;\blambda)$. However, approximate saddle points can be reached under specific settings using gradient-based methods, such as the PD method~\cite{nedic2009}. 
\color{black}

According to the PD method, during iteration $t$, the primal and dual variables are updated as:
\begin{subequations}\label{eq:pd1}
\begin{align}
\bx^{t+1}&\coloneqq \bx^{t}-\mu^{t}_x \nabla_{\bx}\mcL(\bx^t;\blambda^t)\label{eq:pd1:p}, \\
\blambda^{t+1}&\coloneqq \left[\blambda^{t}+\mu^{t}_{\lambda} \nabla_{\blambda}\mcL(\bx^{t+1};\blambda^t)\right]_+,\label{eq:pd1:d}
    \end{align}   
\end{subequations}
where $\mu_x^t$ and $\mu_{\lambda}^t$ are positive step sizes, and $[\lambda]_+ \coloneqq \max\{\lambda,0\}$. The update in \eqref{eq:pd1:p} is a gradient descent step on $\mcL$ in terms of the primal variable. The update in \eqref{eq:pd1:d} is a projected gradient ascent step on $\mcL$ in terms of the dual variable.

One may attempt to solve \eqref{eq:qc} using the PD iterations in \eqref{eq:pd1}. In this case, computing the Lagrangian gradients incurs complexity $\mcO(sN^2)$, assuming each row of $\bM_m$ has at most $s$ nonzero entries. To see this, note that the partial derivative of $\mcL$ with respect to $\lambda_m$ is
\[\frac{\partial \mcL}{\partial\lambda_m}=\bx^\dag \bM_m \bx-b_m,\]
which can be computed in $\mcO(sN)$ operations. Therefore, computing the entire gradient $\nabla_{\blambda}\mcL$ entails complexity $\mcO(sNM)=\mcO(sN^2)$ given that $M$ is proportional to $N$. Similarly, computing the gradient
\[\nabla_{\bx} \mcL=2\left(\bM_0+\sum_{m=1}^M \lambda_m \bM_m\right)\bx\]
incurs $\mcO(sN^2)$ operations. For large-scale problems, such as solving AC OPF for contemporary power systems with more than $N=50,000$~nodes, this computational cost quickly becomes prohibitive. In this context, can a VQA attain reduced computational complexity? The next section proposes a variational model for QCQP/SDPs. 

\section{A doubly variational QCQP/SDP}\label{sec:parameterizedopf}
Variational quantum computing is an algorithmic tool for solving high-dimensional problems in a parameterized form using a PQC. The well-known VQE solves a specific form of~\eqref{eq:qc}~\cite{peruzzo2014}:
\begin{align}
\min_{\bx\in \mathbb{C}^{N}}~&~\bx^\dag \bM_0 \bx\label{eq:vqe}\\
\textrm{s.to}~&~\bx^\dag\bx=1.\nonumber
\end{align}
The optimal cost in \eqref{eq:vqe} is the smallest eigenvalue of~$\bM_0$. Moreover, the eigenvector corresponding to the smallest eigenvalue of $\bM_0$ is the minimizer. The key idea of VQE is to model the original high-dimensional vector~$\bx$ as a quantum state $\ket{\bpsi}$ using $\log N$ qubits. Suppose for now that $N$ is a power of 2. The constraint in \eqref{eq:vqe} can be omitted as the quantum state has unit Euclidean norm anyway. This quantum state is generated upon applying a PQC on an initial state $\ket{\bzero}_{\log N}$. The PQC is modeled by a unitary matrix $\bU(\btheta)$ parameterized by vector $\btheta \in [0,2\pi]^{P}$ with $P \ll N$. Restricting $\btheta$ to $[0,2\pi]^{P}$ is trivial due to the $2\pi$-periodicity of the expectations. The PQC generates the state $\ket{\bpsi(\btheta)}=\bU(\btheta)\ket{\bzero}$, and then measures the expectation $F_0(\btheta)=\braket{\bpsi(\btheta)|\bM_0|\bpsi(\btheta)}$ and its gradient $\nabla_{\btheta} F_0(\btheta)$. Subsequently, a classical computer receives the gradient information and updates $\btheta$ to minimize $F_0(\btheta)$ using gradient descent or other classical optimization methods. 

We propose solving the~\eqref{eq:qc} in a similar manner. Because the primal variable of the~\eqref{eq:qc} does not have a unit Euclidean norm, we parameterize it as~\cite{scaling_state}:
\begin{equation}\label{eq:modelv}
\bx(\btheta)=\alpha \ket{\bpsi(\btheta)}
\end{equation}
where $\alpha>0$ is an auxiliary optimization variable. The PQC used to generate $\ket{\bpsi(\btheta)}$ will be henceforth termed the \emph{primal PQC} or PQC$_p$. 

Given the parameterization in \eqref{eq:modelv}, the problem~\eqref{eq:qc} can be expressed in the variational form:
\begin{align}\label{eq:qcF}\tag{QCQP$_{\theta}$}
P_{\theta}^*\coloneqq \min_{\btheta,\alpha>0}~&~\alpha^2F_0(\btheta)\\ 
\mathrm{s.to}~&~\alpha^2F_m(\btheta)\leq b_m,\quad m=1:M\nonumber
\end{align}
where expectations are defined as
\begin{equation}\label{eq:observable1}
F_m(\btheta)\coloneqq \braket{\bpsi(\btheta)|\bM_m|\bpsi(\btheta)},\quad m=0:M.
\end{equation}
These expectations can be alternatively expressed via the density operator $\brho(\btheta)$ associated with the PQC$_p$ state as
\begin{equation}\label{eq;observable2}
F_m(\btheta)=\trace(\bM_m\brho(\btheta)),\quad m=0:M.
\end{equation}
If PQC$_p$ is characterized by a pure state $\ket{\bpsi(\btheta)}$, its density operator is $\brho(\btheta)=\ket{\bpsi(\btheta)} \!\bra{\bpsi(\btheta)}$. Generally, a density operator is a positive semidefinite (PSD) and unit-trace matrix~\cite{nielsen00}. Under this interpretation,~\eqref{eq:qcF} is a variational SDP over the density operator $\brho(\btheta)$.

As in VQE, the original \eqref{eq:qc} over $\bx$ has been replaced by the variational problem \eqref{eq:qcF} over the parameter vector $\btheta$. Attempting to solve~\eqref{eq:qcF} using Lagrangian duality entails solving the dual problem:
\begin{equation*}
D_{\theta}^*\coloneqq \max_{\blambda\geq \bzero}~\min_{\btheta,\alpha>0}\mcL(\btheta,\alpha;\blambda)
\end{equation*}
over $\blambda\in\mathbb{R}^M_+$. The related Lagrangian function is
\begin{align}\label{eq:Lagrangian1o}
\mcL(\btheta,\alpha;\blambda)&= \alpha^2F_0(\btheta)+\sum_{m=1}^M\lambda_m\left(\alpha^2F_m(\btheta)-b_m\right).
\end{align}
Unfortunately, because the dual variable $\blambda$ has length $\mcO(N)$, the curse of dimensionality remains.

To bypass this difficulty, dual variables can also be represented using a variational model. A similar idea was proposed in~\cite{patel2024} for capturing the primal and dual variables of an SDP. To model $\blambda$ variationally, we introduce a second PQC, henceforth termed the \emph{dual PQC} or PQC$_d$. The dual PQC operates on $\log M$ qubits, supposing again that the number of constraints, $M$, is a power of two. The dual PQC is modeled by a unitary matrix $\bV(\bphi)$ parameterized by vector $\bphi\in [0,2\pi]^Q$ such that $Q\ll M$. The state of PQC$_d$ is $\ket{\bxi(\bphi)}=\bV(\bphi)\ket{\bzero}$; see Fig.~\ref{fig:PD_scheme}. 

With the aid of PQC$_d$, the $m$-th entry of $\blambda$ is parameterized as
\begin{equation}\label{eq:modellambda}
\lambda_m(\bphi)=\beta^2\cdot|\xi_m(\bphi)|^2, \quad\text{for}~m=1:M
\end{equation}
where $\beta>0$ is an auxiliary variable and $\xi_m(\bphi)$ is the $m$-th entry of the state $\ket{\bxi(\bphi)}$. Like $\alpha$ in \eqref{eq:modelv}, variable $\beta$ is introduced for scaling purposes. Unlike \eqref{eq:modelv}, however, dual variables are related to the \emph{squared magnitudes} of quantum state entries. This ensures that $\lambda_m(\bphi)$ takes real nonnegative values. 


The Lagrangian function related to this doubly parameterized QCQP can be expanded into three terms:
\begin{equation}\label{eq:Lagrangian1}
\mcL(\btheta,\alpha;\bphi,\beta)= \alpha^2F_0(\btheta)+\alpha^2\beta^2F(\btheta,\bphi) -\beta^2 G(\bphi),
\end{equation}
where 
\begin{equation}\label{eq:LagrangianF}
F(\btheta,\bphi)\coloneqq \sum_{m=1}^M|\xi_m(\bphi)|^2F_m(\btheta)
\end{equation}
and 
\begin{equation}\label{eq:LagrangianG}
G(\bphi)\coloneqq \sum_{m=1}^M b_m|\xi_m(\bphi)|^2.
\end{equation}

The variational dual problem is defined as
\begin{equation}\label{eq:dualThetaPhi}
D_{\theta,\phi}^*\coloneqq \max_{\bphi,\beta>0}\min_{\btheta,\alpha>0} \mcL(\btheta,\alpha;\bphi,\beta).
\end{equation}

One may wonder why we do not parameterize the primal and dual variables jointly using a single PQC operating on $\log N+\log M$ qubits. Unfortunately, under that design, we would not be able to update the PQC parameters in a meaningful way. Having two PQCs parameterized separately, as depicted in Figure~\ref{fig:PD_scheme}, allows us to minimize the Lagrangian function in \eqref{eq:Lagrangian1} over $(\btheta,\alpha)$ and maximize it over $(\bphi,\beta)$. 

The standard Lagrangian dual function is known to be concave with respect to the dual variables, even if the primal problem is nonconvex. That is not the case for the dual function $\min_{\btheta,\alpha>0} \mcL(\btheta,\alpha;\bphi,\beta)$ in \eqref{eq:dualThetaPhi}. This is because the dual variables are now parameterized, and $\mcL(\btheta,\alpha;\bphi,\beta)$ is not concave in $(\bphi,\beta)$. 

\section{Solving the doubly variational QCQP/SDP}\label{sec:solveparameterizedopf}
\begin{figure}[t]
\centering
\includegraphics[width=1\linewidth]{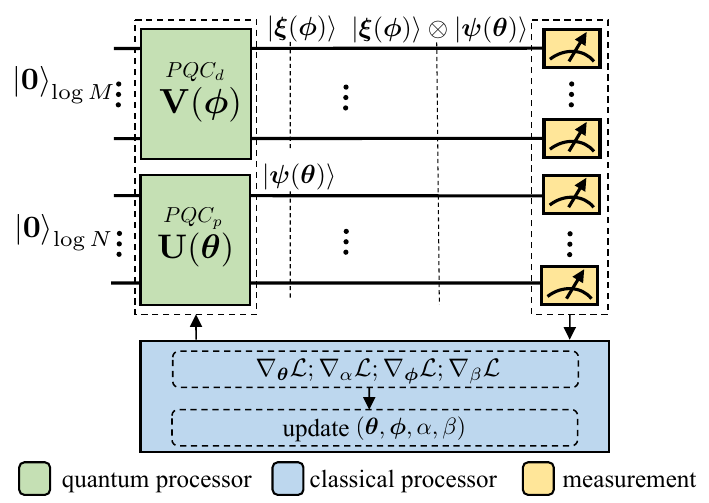}
\caption{Workflow of approximately solving~\eqref{eq:qc} using a doubly variational approach. The primal PQC (bottom) parameterizes the primal variables. The dual PQC (top) parameterizes the dual variables.}
\label{fig:PD_scheme}
\end{figure}

We propose seeking an approximately stationary point of the Lagrangian function in \eqref{eq:Lagrangian1} using a PD algorithm. Let $\bar{\alpha}$ and $\bar{\beta}$ be upper bounds on the optimal values of $\alpha$ and $\beta$, respectively. Such bounds are typically known from the problem specifications; otherwise, these bounds can be set to sufficiently large values. The $t$-th iteration of the algorithm comprises four steps:
\begin{subequations}\label{eq:pdmv}
    \begin{align}
    \btheta^{t+1}&\coloneqq \btheta^{t}-\mu^{t}_{\theta} \nabla_{\btheta}\mcL(\btheta^t,\alpha^t;\bphi^t,\beta^t)\label{eq:pdm:theta},\\
    \alpha^{t+1}&\coloneqq \Pi_{[0,\bar{\alpha}]}\left(\alpha^{t}-\mu^{t}_{\alpha} \nabla_{\alpha}\mcL(\btheta^t,\alpha^t;\bphi^t,\beta^t)\right)\label{eq:pdm:alpha},\\
    \bphi^{t+1}&\coloneqq \bphi^{t}+\mu^{t}_{\phi} \nabla_{\bphi}\mcL(\btheta^t,\alpha^t;\bphi^t,\beta^t)\label{eq:pdm:phi},\\
    \beta^{t+1}&\coloneqq \Pi_{[0,\bar{\beta}]}\left(\beta^{t}+\mu^{t}_{\beta} \nabla_{\beta}\mcL(\btheta^t,\alpha^t;\bphi^t,\beta^t)\right)\label{eq:pdm:beta},
    \end{align}   
\end{subequations}
where $\mu^{t}_{\theta}$, $\mu^{t}_{\alpha}$, $\mu^{t}_{\phi}$, and $\mu^{t}_{\beta}$ are positive step sizes, and $\Pi_{[\cdot]}$ denotes the Euclidean projection onto the interval $[\cdot]$. Steps~\eqref{eq:pdm:theta}--\eqref{eq:pdm:alpha} constitute gradient descent steps to update the primal variables, while~\eqref{eq:pdm:phi}--\eqref{eq:pdm:beta} are gradient ascent steps to update the dual variables. The updates of \eqref{eq:pdmv} run on a classical computer. The suggested workflow is shown in Fig.~\ref{fig:PD_scheme}.  

The gradients in \eqref{eq:pdmv} can be estimated with the help of the two PQCs. The key observation is that each term of the Lagrangian function can be expressed as an expectation over an observable. More specifically, the first term can be computed using an observable operating on PQC$_p$ through the expectation:
\begin{equation}\label{eq:term1}
F_0(\btheta)=\braket{\bpsi(\btheta)|\bM_0|\bpsi(\btheta)}.
\end{equation}
The third term can be computed using an observable operating on PQC$_d$ through the expectation:
\begin{equation}\label{eq:term3}
G(\bphi)=\braket{\bxi(\bphi)|\bS|\bxi(\bphi)},
\end{equation}
where $\bS\coloneqq \diag(\{b_m\}_{m=1}^M)$ is a diagonal matrix with the OPF parameters $b_m$ on its main diagonal. Due to the parameterization in \eqref{eq:modellambda}, the Lagrangian of the doubly parameterized problem is quadratic in $\ket{\bxi(\bphi)}$, although the Lagrangian in \eqref{eq:Lagrangian1o} was linear in $\blambda$. 

The second term of the Lagrangian can be expressed as the expectation over an observable operating \emph{jointly} on PQC$_p$ and PQC$_d$, as shown in the ensuing lemma proved in Appendix~\ref{sec:app:le1}.

\begin{lemma}\label{le:term2}
The second summand in \eqref{eq:LagrangianF} can be computed as the expectation
\begin{equation}\label{eq:term2}
F(\btheta,\bphi)= \braket{\bxi(\bphi),\bpsi(\btheta)|\bM|\bxi(\bphi),\bpsi(\btheta)}
\end{equation}
defined by the $MN\times MN$ Hermitian matrix
\begin{equation}\label{eq:M}
\bM=\sum_{m=1}^M \be_m \be_m^\top \otimes \bM_m
\end{equation}
and $\be_m$ is the $m$-th column of $\bI_M$. 
\end{lemma}

Lemma~\ref{le:term2} asserts that $F(\btheta,\bphi)$ is an expectation applied to the composite state of the two PQCs. This is important as it bypasses the need to measure each one of the numerous constraint expectations $\{F_m(\btheta)\}_{m=1}^M$ of the OPF separately. This powerful feature was first identified in \cite{patel2024} for solving doubly variational SDPs on a quantum computer. In this work, we adopt this idea to the setting of solving QCQPs and SDPs over sparse graphs. We also propose a method for efficiently measuring the observables associated with sparse graphs. We uniquely leverage the underlying graph structure to efficiently measure the three observables using the two PQCs. We defer this discussion to Section~\ref{sec:measure}. If the observables can be measured efficiently, their gradients can be measured efficiently as well via the so-called \emph{parameter-shift rule} (PSR). We review this rule and adapt it to the PD iterations in \eqref{eq:pdmv}.

Let us focus on the gradient $\nabla_{\btheta} F_0(\btheta^t)$. Suppose that the unitary matrix modeling PQC$_p$ takes the form 
\begin{equation}
\bU(\btheta)=\bW_{P+1}\prod_{p=1}^P \exp(-j\theta_p\bG_p)\bW_p,    
\end{equation}
where each $\bG_p$ is a single-qubit Hermitian generator with two distinct eigenvalues $\pm r$, and $\{\bW_p\}_{p=1}^{P+1}$ is a set of fixed gates. Then, the partial derivative of $F_0(\btheta^t)$ with respect to the $p$-th entry of $\btheta$ can be computed by evaluating $F_0(\btheta)$ at two values of $\btheta$~\cite{mitarai2018,schuld2019}:
\begin{equation}\label{eq:psr}
\frac{\partial F_0(\btheta^t)}{\partial \theta_p}=r\left(F_0(\btheta^t+\tfrac{\pi}{4r}\be_p)-F_0(\btheta^t-\tfrac{\pi}{4r}\be_p)\right),
\end{equation}
where $\be_p$ is the $p$-th column of the identity matrix $\bI_P$. If $\bG_p$ is a Pauli rotation from the set $\frac{1}{2} \{\sigma_x,\sigma_y,\sigma_z\}$, we get $r=\frac{1}{2}$ and each $\theta^t_p$ is shifted by $\pm \frac{\pi}{2}$.

The expectation value $F_0(\btheta)$ cannot be evaluated exactly; it can only be estimated via quantum measurements: To estimate $F_0(\btheta)$ for a particular $\btheta$ with precision $\epsilon$, PQC$_p$ has to be sampled $S=\mcO(\epsilon^{-2})$ times. To obtain the complete gradient $\nabla_{\btheta} F_0(\btheta^t)$, the previous process is repeated $2P$ times and involves a total of $2PS$ measurement samples. 

The gradient of the expectation value in \eqref{eq:term3} with respect to $\bphi$ can be computed similarly by running PQC$_d$ by shifting the values of $\bphi^t$ for $2Q$ times or $2QS$ measurement samples. The gradient $\nabla_{\btheta} F(\btheta^t,\bphi^t)$ can be computed by running both PQCs while only shifting the entries of $\btheta^t$; and vice versa for computing the gradient $\nabla_{\bphi} F(\btheta^t,\bphi^t)$. It should be clear by now that, due to the PSR, evaluating gradients of expectations is equivalent to measuring expectations. Section~\ref{sec:measure} discusses how to measure the observables induced from a sparse graph.

Although the plain PD method has been extensively studied for convex/concave and nonconvex/concave saddle-point problems~\cite{kallio1994,uzawa1958,koshal2011,lin2020}, its convergence is not fully understood in the nonconvex/nonconcave setting, like the one in~\eqref{eq:dualThetaPhi}. Recently, a variation of the PD method, termed the extragradient (EG) method~\cite{korpelevich1976}, has been shown to converge under certain conditions for nonconvex/nonconcave saddle-point problems~\cite{minimax}. Compared to the plain PD method, each EG iteration doubles the number of gradient evaluations of the Lagrangian function. More specifically, let vector $\bz^t\coloneqq [\btheta^t;\alpha^t;\bphi^t;\beta^t]$ collect the primal and dual variables at iterate $t$, and define the vector:
\begin{equation}\label{eq:g}
\bg(\bz^t)\coloneqq [\nabla_{\btheta}\mcL(\bz^t);\nabla_{\alpha}\mcL(\bz^t);-\nabla_{\bphi}\mcL(\bz^t);-\nabla_{\beta}\mcL(\bz^t)].
\end{equation}
Then, the EG iterations can be expressed as 
\begin{subequations}\label{eq:eg}
    \begin{align} 
    \bbz^t&=\bz^t-2\mu_z^t\bg(\bz^t)\label{eq:eg:3},\\
    \bz^{t+1}&=\bz^t-\mu_z^t\bg(\bbz^t) \label{eq:eg:4},
    \end{align}
\end{subequations}
for a step size $\mu_z^t>0$. In essence, the EG method first computes an intermediate point $\bbz^t$ based on the plain PD method, and then updates $\bz^t$ to $\bz^{t+1}$ upon evaluating the gradient operator at the intermediate point. Despite doubling the number of gradient evaluations, EG iterates feature favorable convergence properties as established analytically in Section~\ref{sec:analysis} and corroborated numerically in Section~\ref{sec:tests}.

Before concluding this section, it is worth noting that, as with other VQAs, training the primal-dual PQC pair could be impeded by the so-called \emph{barren plateaus} issue, according to which gradients vanish exponentially fast in terms of the number of qubits and depth of PQCs~\cite{mcclean2018,Larocca2025}. This phenomenon can occur when designing VQAs for large-scale problems and/or using overly expressive PQCs. A related concern is that the PQC architectures that are both sufficiently expressive and free from barren plateaus are often efficiently simulated by classical algorithms~\cite{cerezo2023}. Interestingly, recent results show that the so-called dynamic PQCs~\cite{dynamic}, which incorporate early measurements followed by measurement-conditional unitaries, and quantum recurrent embedding neural networks~\cite{jing2025quantumrecurrentembeddingneural} can provably avoid barren plateaus while retaining high expressivity. Moreover, dynamic PQCs inherit the worst-case classical hardness of the universal unitary circuits, precluding the possibility of classical simulatability. Investigating the performance of the doubly variational QCQP/SDP on dynamic PQCs or quantum recurrent embedding neural networks for large-scale problems is an interesting future direction.

\section{Measuring Observables Induced by Sparse Graphs}\label{sec:measure}
This section reviews the extended Bell measurement (XBM) method proposed in \cite{Kondo2022} and explains how it can be adapted to efficiently measure observables induced from sparse graphs.

\subsection{Extended Bell measurement (XBM) method}\label{subsec:xbm}

Consider measuring the expectation value $\bra{\bpsi(\btheta)}\bM_m\ket{\bpsi(\btheta)}$ associated with the $N\times N$ Hermitian matrix $\bM_m$. The matrix~$\bM_m$ is conventionally expressed as
\begin{equation}\label{eq:xbm1}
\bM_m=\sum_{i=0}^{N-1}\sum_{j=0}^{N-1} \braket{i|\bM_m|j} \ket{i}\!\bra{j}.
\end{equation}
The XBM method proposes the alternative expansion:
\begin{equation}\label{eq:xbm2}
\bM_m=\sum_{c=0}^{N-1}\underbrace{\sum_{i=0}^{N-1} \braket{i|\bM_m|i\oplus c} \ket{i}\!\bra{i\oplus c}}_{\coloneqq \bM_m^c}=\sum_{c=0}^{N-1}\bM_m^c,
\end{equation}
where $\oplus$ denotes the bitwise exclusive OR (XOR) operator if indices $i$ and $c$ are represented in binary form. To establish the equivalence between \eqref{eq:xbm1} and \eqref{eq:xbm2}, we need to show that every pair $(i,j)$ corresponds to a unique pair $(i,i\oplus c)$, which indeed holds for $c=i\oplus j$. Due to this alternative indexing, the matrix $\bM_m$ can be decomposed into $N$ matrices $\{\bM_m^c\}_{c=0}^{N-1}$ defined as
\[\braket{i|\bM_m^c|j}=
\begin{cases}
\braket{i|\bM_m|j}, &\text{if}~i\oplus j =c\\ 
0, & \text{otherwise}.
\end{cases}
\]
For example, the matrix $\bM_m^0$ carries the diagonal entries of $\bM_m$ because $i\oplus i=0$. Figure~\ref{fig:colorMatrix} uses color coding to identify the index pairs $(i,j)$ corresponding to the same $c$ for an observable over three qubits~\cite{Kondo2022}. We will henceforth say that all index pairs $(i,j)$ satisfying $i\oplus j =c$ for a specific $c$ belong to the same \emph{color}~$c$.

\begin{figure}[t]
\centering
\includegraphics[width=0.8\linewidth]{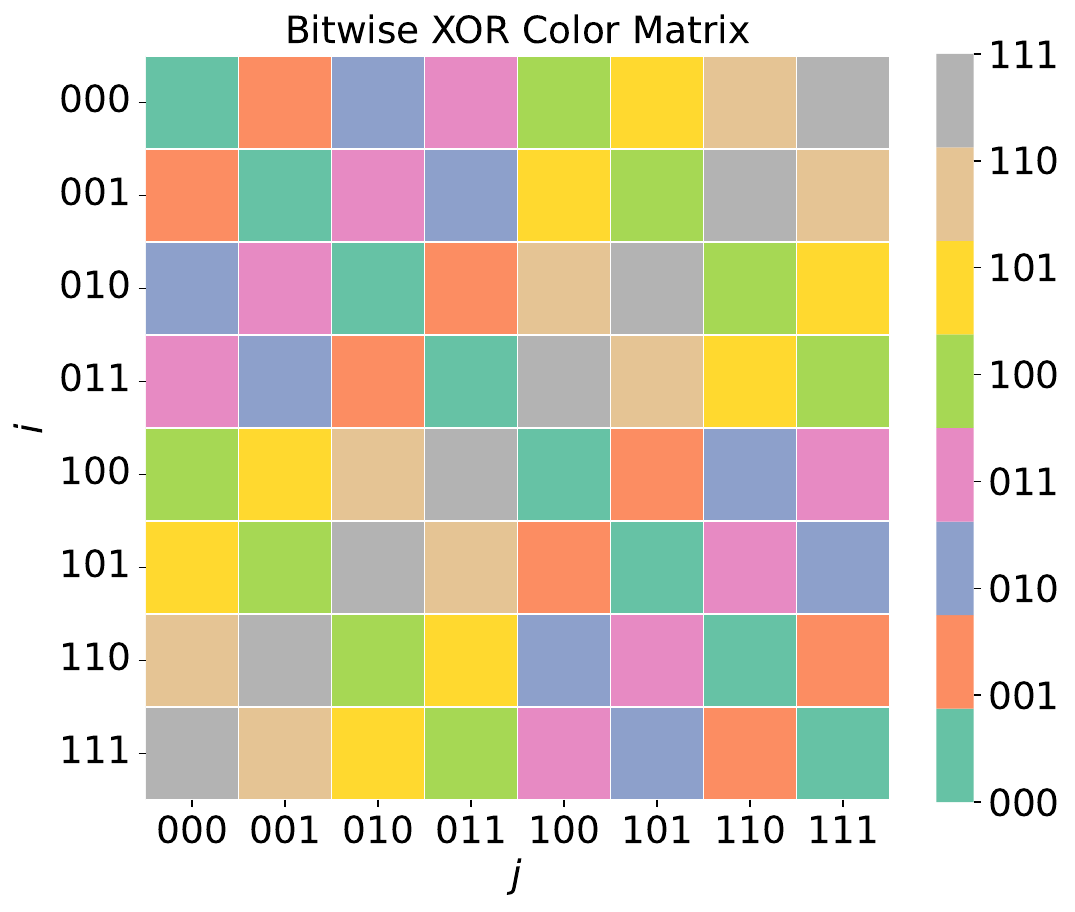}
\caption{The XBM method of \cite{Kondo2022} groups the entries of an $N\times N$ Hermitian matrix $\bM_m$ into $N$ groups or colors. This XBM grouping or entry coloring is shown here for an $8\times 8$ matrix associated with $\log 8=3$ qubits. Each color $c$ is identified by the binary form of its index $c$. Color $c$ includes all matrix entries with index pairs $(i,j)$ for which $i\oplus j =c$. For example, the grey color $c=7=\ket{111}$ includes all entries lying on the main anti-diagonal of $\bM_m$. Note that the color arrangement is symmetric.}
\label{fig:colorMatrix}
\end{figure}

Because $i\oplus j = j\oplus i$, every $\bM_m^c$ is Hermitian. Therefore, the matrices
\begin{equation}\label{eq:xbm3}
\hbM_m^c\coloneqq \real\{\bM_m^c\}\quad\text{and}\quad\cbM_m^c\coloneqq \iota\imag\{\bM_m^c\}
\end{equation}
are also Hermitian for all $c$, and $\cbM_m^0=\bzero$. 

Based on~\eqref{eq:xbm2}, the expectation value associated with $\bM_m$ can be expressed as a sum of $2N-1$ observables:
\begin{equation}\label{eq:xbm4}
\braket{\bpsi|\bM_m|\bpsi}=\sum_{c=0}^{N-1}\braket{\bpsi|\hbM_m^c|\bpsi}+ \sum_{c=1}^{N-1}\braket{\bpsi|\cbM_m^c|\bpsi}.
\end{equation}
Because $\hbM_m^0$ is diagonal, the expectation value $\braket{\bpsi|\hbM_m^0|\bpsi}$ can be measured in the computational basis. Upon measuring $\ket{\bpsi}$, the binary vector outcome~$\ket{i}$ is observed with probability $|\psi_i|^2$. Define a random variable taking the value $\braket{i|\hbM_m^0|i}$ when outcome $\ket{i}$ is observed. Then, the expectation $\braket{\bpsi|\hbM_m^0|\bpsi}$ in~\eqref{eq:xbm4} is the expected value of this random variable.   

The remaining expectation values in \eqref{eq:xbm4} are not as straightforward to measure because the matrices $(\hbM_m^c,\cbM_m^c)$ are non-diagonal for $c=1,\ldots,N-1$. The advantage of XBM's color decomposition is that these matrices can be diagonalized by unitaries corresponding to qubit-efficient circuits. This key result from \cite{Kondo2022} is summarized in the following lemma.

\begin{lemma}[\cite{Kondo2022}]\label{le:kondo}
Consider the color index $c\geq 1$. Let index $k_c$ denote the most significant bit taking the value of one in the binary representation of $\ket{c}=\ket{c_{\log N-1}\cdots c_1 c_0}$. Then, the matrices $\hbM_m^c$ and $\cbM_m^c$ admit the eigenvalue decompositions:
\begin{equation*}
\hbM_m^c=\bU_c \hbLambda_m^c \bU_c^\top 
\quad \text{and} \quad
\cbM_m^c=\bS_{c}^\dag\bU_c \cbLambda_m^c \bU_c^\top\bS_{c}
\end{equation*} 
where
\begin{itemize}
\item[{i)}] unitary matrix $\bU_c$ can be implemented using one Hadamard gate and at most $(\log N-1)$ CNOT gates;
\item[{ii)}] unitary matrix $\bS_{c}$ applies a phase gate $\bS$ on qubit~$k_c$ and acts trivially on all other qubits; and 
\item[{iii)}] the diagonal eigenvalue matrices are
\begin{align*}
\hbLambda_m^c&\coloneqq \frac{1}{2}\sum_{i=0}^{N-1}\braket{i| \hbM_m^c |i\oplus c}\left(\ket{i}\!\bra{i}-\bX_{k_c}\ket{i}\!\bra{i}\bX_{k_c}^\top\right)\\
\cbLambda_m^c&\coloneqq \frac{-\iota}{2}\sum_{i=0}^{N-1}\braket{i| \cbM_m^c |i\oplus c}\left(\ket{i}\!\bra{i}-\bX_{k_c}\ket{i}\!\bra{i}\bX_{k_c}^\top\right)
\end{align*}
with $\bX_{k_c}$ applying a NOT gate on qubit $k_c$.
\end{itemize}
\end{lemma}

The practical value of Lemma~\ref{le:kondo} is that the expectation value
\[\braket{\bpsi| \hbM_m^c| \bpsi}=\braket{\bpsi| \bU_c \hbLambda_m^c \bU_c^\top| \bpsi}\]
can be measured by transforming $\ket{\bpsi}$ to $\bU_c^\top\ket{\bpsi}$ and then measuring the latter in the computational basis using the diagonal matrix $\hbLambda_m^c$. The transformation is qubit-efficient because the unitary $\bU_c$ involves at most $\log N$ elementary gates. The expectation value $\braket{\bpsi| \cbM_m^c| \bpsi}$ can be measured similarly, with at most $\log N+1$ elementary gates; $\log N$ for $\bU_c$ plus one for $\bS_{c}$. A critical feature is that the unitary $\bU_c$ diagonalizing $(\hbM_m^c,\cbM_m^c)$ depends solely on the color (locations of nonzero entries) of $(\hbM_m^c,\cbM_m^c)$, and \emph{not} on the values of those entries. Only the eigenvalue matrices $(\hbLambda_m^c,\cbLambda_m^c)$ depend on those values, and can be computed without the need for an eigenvalue decomposition. 

Although each expectation value in \eqref{eq:xbm4} can be estimated efficiently, there are $(2N-1)$ of them. Unfortunately, each one of them requires a different quantum circuit. Therefore, the expectation value $\braket{\bpsi| \bM_m|\bpsi}$ can be measured efficiently only if all nonzero entries of $\bM_m$ are covered by a few colors. Let $\mcC$ denote the set of colors for which either $\hbM_m^c$ or $\cbM_m^c$ is nonzero, so that $\bM_m$ is decomposed as $\bM_m=\sum_{c\in\mcC}\bM_m^c$. For the XBM protocol to be qubit-efficient, the number of colors $C=|\mcC|$ should scale polynomially in $\log N$. Reference~\cite{Kondo2022} establishes that \emph{banded} matrices occupy relatively few colors. Specifically, it is shown that a $k$-banded matrix can be decomposed using $C=\mcO(k\log N)$ rather than $N$ colors. For the example of Figure~\ref{fig:colorMatrix}, a matrix with bandwidth $k=1$ requires 4 colors. 

Based on Lemma~\ref{le:kondo}, if the expectation value $F_m(\btheta)$ occupies $C$ colors, it can be decomposed into $2C-1$ expectations:
\begin{equation}\label{eq:xbm5}
F_m(\btheta)=\sum_{c=1}^{2C-1}F_m^c(\btheta)=\sum_{c=1}^{2C-1}\braket{\bpsi_c(\btheta)|\bLambda_m^c|\bpsi_c(\btheta)},
\end{equation}
where $\ket{\bpsi_c(\btheta)}=\bU_c^\top\ket{\bpsi(\btheta)}$ and $\bLambda_m^c=\hbLambda_m^c$ for $c=1,\ldots,C$; and $\ket{\bpsi_c(\btheta)}=\bU_c^\top\bS_c\ket{\bpsi(\btheta)}$ and $\bLambda_m^c=\cbLambda_m^c$ for $c=C+1,\ldots,2C-1$. Each one of the observables involved can be diagonalized by a different unitary matrix. We next describe how to effect a measurement matrix $\bM$ that uses a few colors.

\subsection{Node permutation for qubit-efficient measurements}\label{subsec:xbm4ps}
To solve the doubly parameterized QCQP/SDP, we need to evaluate the Lagrangian function in \eqref{eq:Lagrangian1} by measuring the observable in Lemma~\ref{le:term2} so that we can compute its gradients over $(\btheta,\alpha;\bphi,\beta)$. We next adopt the XBM protocol to efficiently measure the related observables in~\eqref{eq:Lagrangian1}.

According to \eqref{eq:M}, observable matrix $\bM$ has the $N\times N$ matrices $\{\bM_m\}_{m=1}^M$ as its diagonal blocks. Despite being block-diagonal, matrix $\bM$ may still occupy up to $N$ colors. Critically, we do not necessarily have to work with $\bM$. Instead, we can permute variables to effect a new measurement matrix that is amenable to more efficient measurements. Suppose we permute vector $\bx$ by a permutation matrix $\bP$ to get $\bP\bx$. This entails reordering the nodes of the underlying graph, which is a trivial task that can be performed before solving~\eqref{eq:qc}. The observable in Lemma~\ref{le:term2} can be expressed as
\begin{equation}\label{eq:perm1}
\sum_{m=1}^M \lambda_m \bx^\dag \bM_m \bx=\sum_{m=1}^M\lambda_m\bx^\dag \bP^\top(\bP\bM_m\bP^\top) \bP\bx.
\end{equation}

Ideally, we would like to design a \emph{single} $\bP$ to minimize the maximum bandwidth across all $\bP\bM_m\bP^\top$. For many conic programs defined over sparse graphs, the constraint matrices $\bM_m$ are dictated by a common underlying graph. In such problems, there often exists a representative matrix $\bY$, for instance, the graph Laplacian matrix, whose sparsity pattern captures the \emph{union} of the sparsity patterns of all $\bM_m$. Therefore, we propose designing $\bP$ to minimize the bandwidth of $\bP\bY\bP^\top$. This is because minimizing the bandwidth of $\bP\bY\bP^\top$ is equivalent to minimizing the maximum bandwidth across all $\bP\bM_m\bP^\top$. In the OPF problem (detailed in Section~\ref{sec:model}), the representative matrix is the node admittance matrix $\bY$ of the power network. We henceforth assume such a representative matrix $\bY$ is available or can be constructed for the problem at hand.

Unfortunately, minimizing the bandwidth of a matrix is an NP-hard problem even if the matrix is sparse~\cite{feige2000}. Nonetheless, heuristic approaches, such as the reverse Cuthill-McKee (RCM) algorithm~\cite{rcm}, provide reasonably good solutions in linear complexity $\mcO(L_e)=\mcO(N)$, where $L_e$ is the number of edges in the underlying graph. The RCM algorithm has been used in classical computing to reduce the bandwidth of sparse matrices arising in various domains, including power systems, where it is applied to the node admittance matrix $\bY$~\cite{rcm_for_powergrids1,rcm_for_powergrids2,rcm_for_powergrids3}. Albeit permuting $\bM_m$'s using the RCM algorithm takes $\mcO(N)$ operations, the crucial point is that this permutation can be performed once for a given graph topology. In applications where the graph structure remains fixed while other problem parameters vary, the same permutation can be reused across multiple problem instances. For example, in the OPF problem, the power system topology changes infrequently and affects only a few edges of the graph. As long as the network topology remains unchanged, an instance of the OPF is characterized by parameters $b_m$, which vary frequently. In contrast to the permutation step, OPF instances over the same topology must be solved repeatedly, every few minutes during real-time operation and potentially millions of times in long-term planning.

\begin{figure}[t]
\centering
\includegraphics[width=0.9
\linewidth]{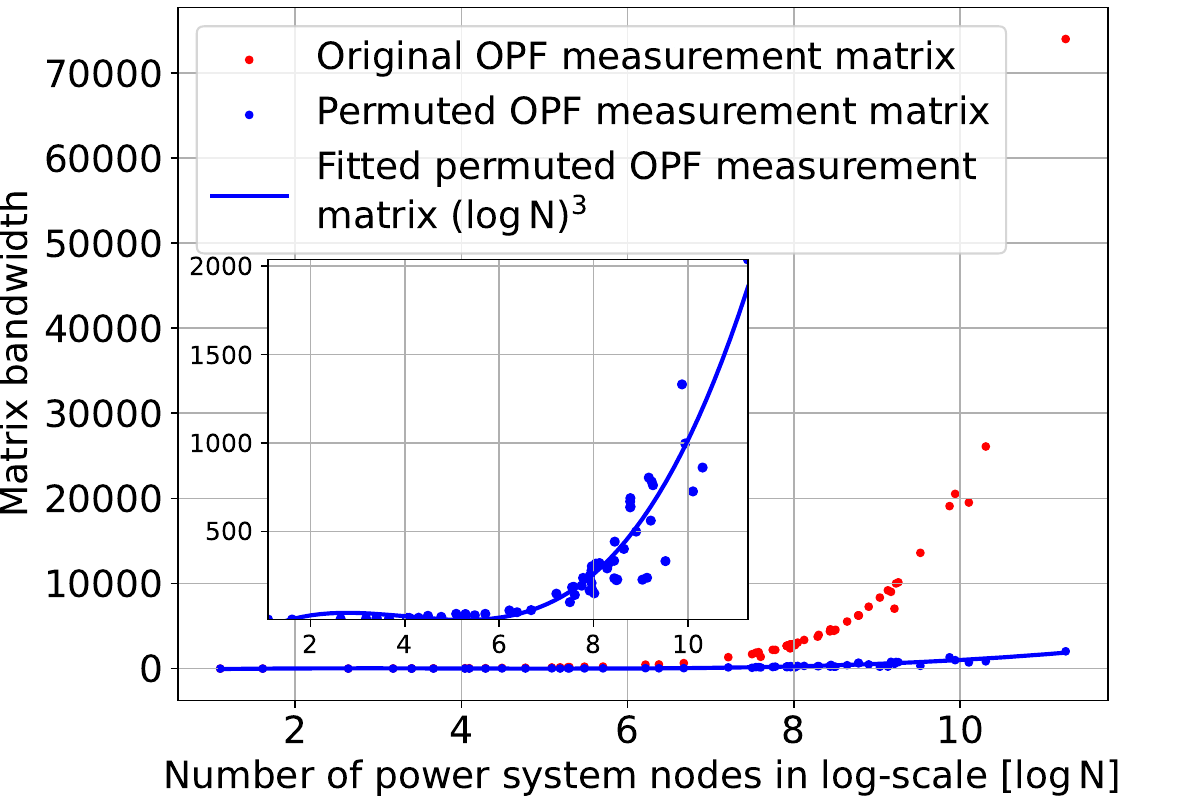}
\caption{The bandwidth of the original and permuted measurement matrices as a function of the logarithmic network size $(\log N)$ for the power system graphs in \texttt{pglib}. Nodes were permuted by the RCM algorithm. Permuted matrices feature patently smaller bandwidths. Upon data fitting and 5-fold cross-validation over different polynomial and exponential functions, the reduced bandwidth was numerically found to scale as $(\log N)^3$. The run time of the RCM algorithm (implemented by function \texttt{scipy.sparse.csgraph.reverse\_cuthill\_mckee} in Python 3.11.9) on the largest power system graph of $N=78,784$ is a few seconds using a MacBook laptop equipped with an M3 Pro processor and 36 GB of RAM.}
\label{fig:bw}
\end{figure}

\begin{figure}[t]
\centering
\includegraphics[width=0.9\linewidth]{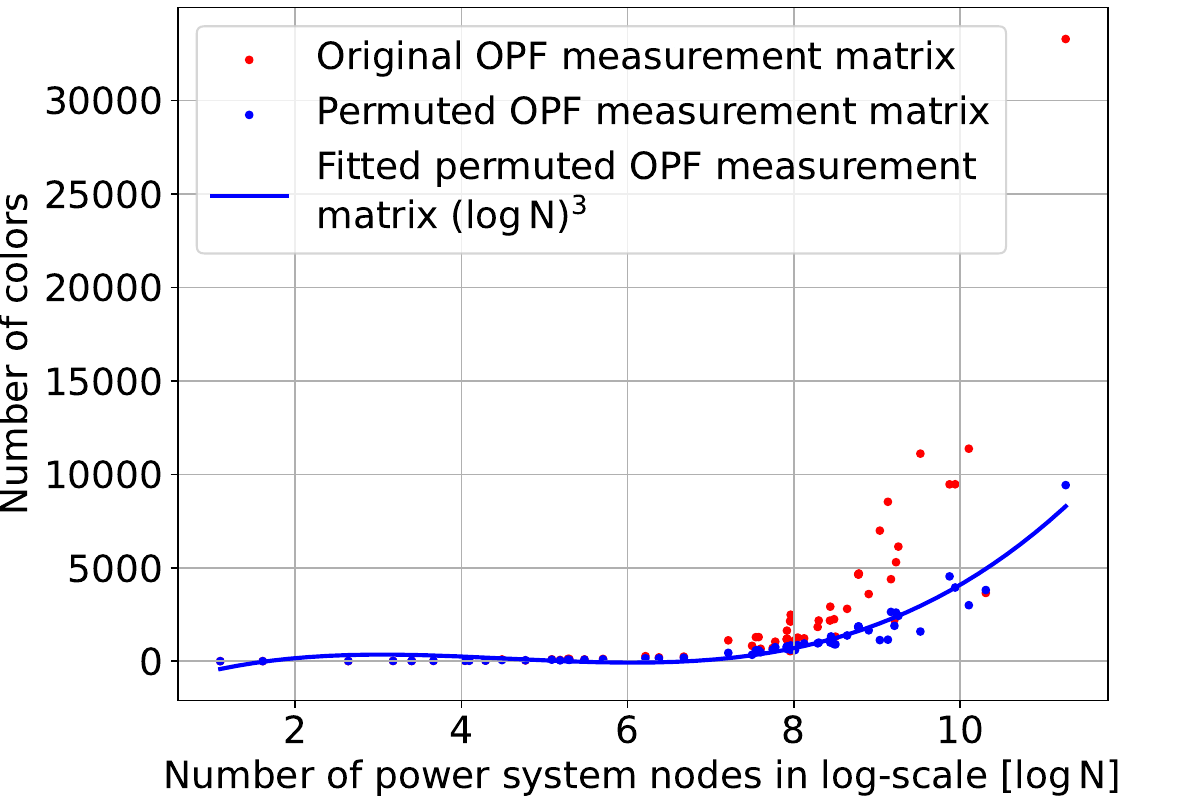}
\caption{The number of colors in the original and permuted measurement matrices as a function of the logarithmic network size $(\log N)$ for the power system graphs in \texttt{pglib}. Permuting nodes using the RCM algorithm was numerically found to significantly reduce the number of colors in use. Data fitting and 5-fold cross validation showed that the number of colors scales as $(\log N)^3$.}
\label{fig:color}
\end{figure}

We numerically validated the effect of node permutation on the bandwidth of $\bM$. We ran the RCM algorithm on the 66 benchmark power system models of the \texttt{pglib} dataset~\cite{pglib}. Figure~\ref{fig:bw} shows the bandwidth of $\bM$ before and after node permutation. Because the RCM solution depends on the initial node, we conducted 200 Monte Carlo runs per power system, with the initial node drawn uniformly at random per run. Out of the 200 runs, we retained the permutation that yielded the smallest bandwidth. Evidently, permuting nodes reduces the bandwidth of $\bM$ dramatically. Upon data fitting and cross-validation, the reduced bandwidth was numerically found to scale as $(\log N)^3$. We also empirically tested the effect of node permutation on the number of colors in $\bM$ for these power systems. Figure~\ref{fig:color} shows that node permutation significantly reduces the number of colors, too. If the bandwidth of $\bP\bM\bP^\top$ is $k=\mcO((\log N)^3)$ per data fitting in Figure~\ref{fig:bw}, the analysis of \cite{Kondo2022} upper bounds the number of its colors as $C=\mcO(k\log N)=\mcO((\log N)^4)$. However, the number of colors in $\bP\bM\bP^\top$ was found to scale as $(\log N)^3$.

We have so far focused on measuring the second observable of the Lagrangian function in \eqref{eq:Lagrangian1}. The third observable is diagonal and is thus straightforward to measure. The first observable $F_0(\btheta)$, however, is defined over the non-diagonal matrix $\bM_0$. 
Fortunately, the sparsity pattern of $\bM_0$ is a subset of the sparsity pattern of $\bY$. Therefore, if nodes are permuted so that $\bM$ occupies a few colors $C$, the same permutation works well also for $\bM_0$. In other words, matrix $\bP\bM_0\bP^\top$ would occupy $C$ colors or fewer.

Since permuting the underlying graph can yield measurement matrices with much fewer colors, we suggest the following workflow for solving doubly parameterized QCQPs over sparse graphs: \emph {i)} Identify the underlying graph and construct a representative matrix $\bY$ whose sparsity pattern is the union of the sparsity patterns of all $\bM_m$. For the OPF, $\bY$ is the node admittance matrix; \emph{ii)} Feed the sparsity pattern of $\bY$ into the RCM algorithm to find a near-optimal node permutation matrix $\bP$; \emph{iii)} Formulate~\eqref{eq:qc} based on the permuted nodes; \emph{iv)} Solve~\eqref{eq:qc} using the VQA of Section~\ref{sec:parameterizedopf}; and \emph{v)} If needed, apply the reverse permutation to the primal solution of~\eqref{eq:qc} to recover the original node ordering. It is henceforth assumed that primal variables and measurement matrices have been permuted already. Therefore, we will henceforth use the lighter notation $\bM$ and $\bM_m$ instead of $\bP\bM\bP^\top$ and $\bP\bM_m\bP^\top$.

\subsection{Implementation details}\label{subsec:implement}

\begin{figure}[t]
\centering
\includegraphics[width=0.85
\linewidth]{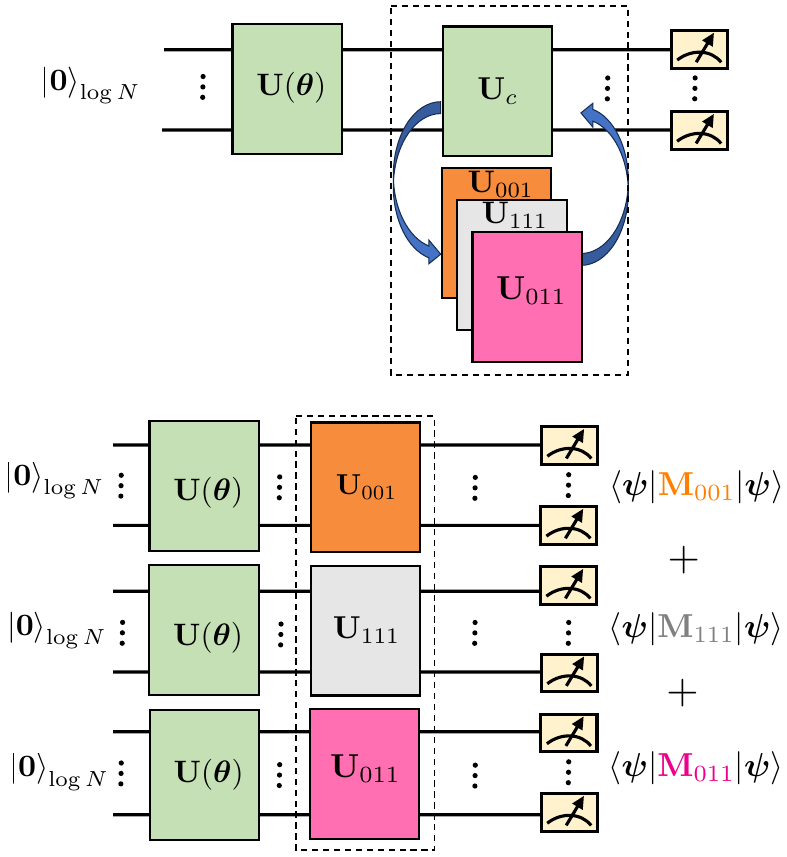}
\caption{Implementation of rotation unitaries. \emph{Top:} A unitary $\bU_c$ is recompiled $2C-1$ times to consider for $2C-1$ rotation unitaries running in sequence. \emph{Bottom:} $2C-1$ rotation unitaries running in parallel.}
\label{fig:rotation}
\end{figure}

We next explain how to estimate the Lagrangian gradients needed in \eqref{eq:pdmv} using the primal/dual PQC pair. We start with $\nabla_{\alpha}\mcL$ in~\eqref{eq:pdm:alpha}, which entails estimating $F_0(\btheta)$ and $F(\btheta, \bphi)$. 
To simplify the presentation of measuring $F_0(\btheta)$, suppose $\bM_0$ occupies all $C$ colors used in $\bY$. According to \eqref{eq:xbm5}, measuring $F_0(\btheta)$ requires $2C-1$ circuits. Each rotated primal circuit $c$ involves the same primal PQC $\bU(\btheta)$ followed by a different XBM unitary $\bU_c$. These $2C-1$ circuits can be executed sequentially or in parallel, as shown in Figure~\ref{fig:rotation}. The sequential implementation operates on $\log N$ qubits and requires compiling $2C-1$ circuits in a sequential fashion. The parallel implementation is $2C-1$ times faster but operates on $\log N(2C-1)$ qubits and requires replicating the primal PQC $2C-1$ times.

Measuring observable $F(\btheta,\bphi)$ is more complicated as it involves estimating the double summation
\begin{align}\label{eq:term2:measure}  
F(\btheta,\bphi)&=\sum_{m=1}^M\sum_{c=1}^{2C-1} |\xi_m(\bphi)|^2F^c_m(\btheta)\nonumber\\
&=\sum_{c=1}^{2C-1}\sum_{m=1}^M |\xi_m(\bphi)|^2F^c_m(\btheta).
\end{align}
Interestingly, for a specific index $c$, the corresponding sum $\sum_{m=1}^M |\xi_m(\bphi)|^2F^c_m(\btheta)$ can be evaluated in two steps. We first measure the dual PQC and observe $\ket{m}$ with probability $|\xi_m(\bphi)|^2$. Given the observed $\ket{m}$, we can now easily estimate the expectation $F_m^c(\btheta)$ corresponding to a diagonal observable. 

Measuring $\nabla_{\beta}\mcL$ in~\eqref{eq:pdm:beta} is trivial as it depends on $F(\btheta, \bphi)$ and $G(\bphi)$.

The gradient $\nabla_{\btheta}\mcL(\btheta)$ can be estimated based on the PSR. For example, the partial derivative of $\mcL$ with respect to $\theta_p$ can be found as
\begin{multline}\label{eq:grad_theta}
\frac{\partial \mcL}{\partial \theta_p}=\frac{\alpha^2}{2}\left(F_0(\btheta^+_p)-F_0(\btheta^-_p)\right) \\
    +\frac{\alpha^2\beta^2}{2}\left(F(\btheta^+_p,\bphi)-F(\btheta^-_p,\bphi)\right)
\end{multline}
where $\btheta^{\pm}_p\coloneqq \btheta \pm \frac{\pi}{2}\be_p$. According to~\eqref{eq:grad_theta}, estimating $\nabla_{\btheta}\mcL$ requires running $2C-1$ rotated primal PQCs, $2P$ times each for $\{\btheta^+_p,\btheta^-_p\}_{p=1}^P$, while the dual PQC is set at $\bphi$. 

Similarly, evaluating $\nabla_{\bphi}\mcL$ in~\eqref{eq:pdm:phi} requires running $2C-1$ rotated primal circuits set at $\btheta$, and the dual PQC for $2Q$ values for $\{\bphi^+_q,\bphi^-_q\}_{q=1}^Q$ as 
\begin{multline}\label{eq:grad_phi}
\frac{\partial \mcL}{\partial \phi_q}=\frac{\alpha^2\beta^2}{2}\left(F(\btheta,\bphi^+_q)-F(\btheta,\bphi^-_q)\right) \\
    -\frac{\beta^2}{2}\left(G(\bphi^+_q)-G(\bphi^-_q)\right)
\end{multline}
where $\bphi^{\pm}_q\coloneqq \bphi \pm \frac{\pi}{2}\be_q$.

Altogether, each PD iteration requires running $\left(2C-1\right)\left(2P+1\right)$ rotated primal circuits and $2Q+1$ dual circuits. Each circuit is run for $S$ samples. According to~\cite{peruzzo2014}, the number of PQC parameters $(P,Q)$ is expected to be polynomial in $\log N$. If the number of colors grows approximately as $C=\mcO((\log N)^3)$, each PD iteration in~\eqref{eq:pdmv} thus incurs a computational complexity of polynomial in $\log N$, which is more efficient compared to $\mcO(sN^2)$ of the PD iteration in~\eqref{eq:pd1} running on a classical computer.

\section{Convergence Analysis}\label{sec:analysis}
This section studies the convergence of the EG iterates in~\eqref{eq:eg}. For a general optimization problem, if the Lagrangian function is convex in the primal variables, concave in the dual variables, and has a Lipschitz gradient in primal and dual variables, the EG iterates converge to a saddle point of the Lagrangian function; see~\cite{korpelevich1976}. For nonconvex/nonconcave saddle-point problems, however, establishing convergence of the EG iterates is challenging. Reference~\cite{minimax} has recently shown that, under certain conditions, the EG iterates converge to an approximately stationary point of the Lagrangian function. Of course, a stationary point may not necessarily be a saddle point of the Lagrangian. Similar convergence claims have also been established in the stochastic setting, where EG updates rely on noisy yet unbiased gradient estimates in lieu of the actual Lagrangian gradients. Given that the Lagrangian function in~\eqref{eq:dualThetaPhi} is nonconvex in $(\btheta,\alpha)$ and nonconcave in $(\bphi,\beta)$, we adopt the convergence proof from~\cite{minimax} for the EG iterates in~\eqref{eq:eg}  and provide an upper bound on the number of measurement samples across all EG iterations. To this end, the following lemma shows that the operator $\bg(\bz)$ is $ L$-Lipschitz continuous; see Appendix~\ref{sec:app:le3} for a proof.

\begin{lemma}\label{le:lipschitz}
The vector-valued function $\bg(\bz)$ defined in~\eqref{eq:g} is Lipschitz continuous with the Lipschitz constant 
\begin{align}\label{eq:lipschitz}&L=\left(P\bar{\alpha}^2\bar{\beta}^2+Q\bar{\alpha}^2\bar{\beta}^2+2\bar{\alpha}\bar{\beta}^2+2\bar{\alpha}^2\bar{\beta}\right)\max_{m}\left\|\bM_m\right\|\nonumber\\
&+\max\{(P\bar{\alpha}^2+2\bar{\alpha})\left\|\bM_0\right\|,(Q\bar{\beta}^2+2\bar{\beta}) \max_m \left|b_m\right|\}.
\end{align}  
\end{lemma}

Let $\mathcal{Z} = [0,2\pi]^P \times [0,\bar{\alpha}] \times [0,2\pi]^Q \times [0,\bar{\beta}]$ denote the compact set for the iterates $\bz^t$. The convergence analysis relies on the ensuing assumption~\cite{minimax}.

\begin{assumption}\label{as:weak_mvi}
There exists a stationary point $\bz^*$ of $\mcL(\bz)$ such that:
    \begin{equation}\label{eq:weak_mvi0}
        \bg(\bz)^\top(\bz-\bz^*)\geq -\frac{\rho}{2}\left\|\bg(\bz)\right\|^2_2,
    \end{equation}
    for all $\bz \in \mcZ$ and some parameter $\rho \in \left[0,\frac{1}{4\sqrt{2}L}\right)$, where $L$ is defined in~\eqref{eq:lipschitz}.
\end{assumption}

Verifying Assumption 1 for high-dimensional nonconvex nonconcave problems such as the one in~\eqref{eq:Lagrangian1} is computationally intractable. Indeed, reference~\cite{anagnostides2025polynomial} recently proved that deciding the classical Minty condition (a special case of Assumption~\ref{as:weak_mvi} for $\rho=0$) is a coNP‑complete problem. In practical terms, this means that if the condition is violated, finding a specific violating point is computationally hard. Moreover, if the condition holds, proving that no such violation exists is also hard. In either case, determining whether the condition holds requires exponential time. Since Assumption~\ref{as:weak_mvi} generalizes the classical Minty condition, verifying the former is at least as hard as the latter. Crucially, Assumption~\ref{as:weak_mvi} is only a sufficient condition for the convergence proof of the extragradient method. Our numerical tests in Section~\ref{sec:tests} show that the proposed doubly variational quantum algorithm can find high‑quality solutions.

Another challenge for our problem is that the EG updates in~\eqref{eq:eg} rely on noisy estimates of $\bg(\bz)$ because observables are measured only based on samples. Specifically, as discussed earlier, each PD iteration in~\eqref{eq:pdmv} utilizes $S$ measurement samples per rotated primal/dual circuit to evaluate the Lagrangian gradients. Similarly, each step~\eqref{eq:eg:3} and~\eqref{eq:eg:4} of the EG updates requires $S$ measurement samples per circuit to estimate $\bg(\bz)$ as:
\begin{equation}
\hbg(\bz)\coloneqq \frac{1}{S}\sum_{s=1}^S  \hbg_s(\bz),
\end{equation}
where $\hbg_s(\bz)$ denotes the estimate of $\bg(\bz)$ based on the single sample $s$. Under this stochastic setting, we adapt the result from~\cite{minimax} to bound the number of iterations $T$ and measurement samples $S$ such that
\begin{equation}\label{eq:expectation1}
\frac{1}{T}\sum_{t=0}^{T-1}\mathbb{E}\!\left[\left\|\hbg(\bar{\bz}^t)\right\|_2\right] \leq \epsilon.
\end{equation}
The expectation in~\eqref{eq:expectation1} is taken over the randomness of $\hbg(\bz)$. Suppose now that the final output $\hbz$ of the EG method is drawn by uniformly sampling at random from the EG iterates $\{\bar{\bz}^t\}_{t=0}^{T-1}$. Then, it is easy to see that $\hbz$ satisfies 
\begin{equation}\label{eq:expectation2}
    \mathbb{E}\!\left[\left\|\hbg(\hbz)\right\|_2\right] \leq \epsilon,
\end{equation}
where the expectation is now taken over the noisy gradients as well as the sampling process to draw the final output $\hbz$. The convergence result is formalized in the following theorem from \cite{minimax}.

\begin{theorem}[\cite{minimax}]\label{th:minimax}
Let $\bg(\bz)$ be an $L$-Lipschitz operator that satisfies Assumption~\ref{as:weak_mvi}. Let $\hbg(\bz)$ be an unbiased estimator of $\bg(\bz)$ such that the variance of $\hbg(\bz)$ is bounded as
 \begin{align}\label{eq:sigma}
    &\mathbb{E}\left[\left\|\hbg(\bz)-\bg(\bz)\right\|_2^2\right]\leq \frac{\sigma^2}{S},
\end{align}
where $S$ is the number of measurement samples per rotated primal/dual circuit.
Given an arbitrary initial point $\bz^0$, run the EG iterations in~\eqref{eq:eg} for $T$ times with the step size $\mu^t_z=\frac{1}{2\sqrt{2}L}$. Accordingly, let $\{\bz^t\}_{t=1}^T$ and $\{\bbz^t\}_{t=0}^{T-1}$ denote the sequences of points generated by the EG iteration. Select a point $\hbz$ uniformly at random from $\{\bar{\bz}^t\}_{t=0}^{T-1}$. 
If the number of iterations is selected as
\begin{equation}\label{eq:num_iter}
T=\left\lceil \frac{32 L^2\left\|\bz^*-\bz^0\right\|_2^2}{\epsilon^2(1-4\sqrt{2}L\rho)}\right \rceil
\end{equation}
and each step~\eqref{eq:eg:3} and~\eqref{eq:eg:4} uses 
\begin{equation}\label{eq:S_one}
S=\left\lceil \frac{8\sigma^2(8+\sqrt{2}L\rho)}{\epsilon^2(1-4\sqrt{2}L\rho)} \right\rceil
\end{equation}
samples per iteration, then $\mathbb{E}\!\left[\left\|\hbg(\hbz)\right\|_2\right]\leq \epsilon$.
\end{theorem}

We next expound on whether the Lagrangian function $\mcL(\bz)$ satisfies the requirements of Theorem~\ref{th:minimax}. First, we were unable to verify Assumption~\ref{as:weak_mvi}, so it is taken as given. Second, the $L$-Lipschitz continuity of the operator $\bg(\bz)$ has been shown in Lemma~\ref{le:lipschitz}. Third, the estimator $\hbg(\bz)$ is unbiased because Lagrangian gradients $\mcL(\bz)$ are computed via the PSR, and estimates of quantum observables based on sample averages are known to be unbiased estimators~\cite{harrow2021}. Lastly, the bound on the variance in~\eqref{eq:sigma} is established in the ensuing lemma shown in Appendix~\ref{sec:app:th}.

\begin{lemma}\label{le:var}
The parameter $\sigma^2$ upper bounding the variance of the gradient estimator in~\eqref{eq:sigma} can be found as:
\begin{align}\label{eq:sigma:1}
&\sigma^2=\frac{Q\bar{\beta}^4+8\bar{\beta}^2}{2}\max_{m}b_m^2+\frac{8\bar{\alpha}^2+P\bar{\alpha}^4}{2}\sum_{c=1}^{2C-1}\|\bM^c_0\|^2\nonumber\\
&+\frac{8\bar{\alpha}^2 \bar{\beta}^4+8\bar{\alpha}^4\bar{\beta}^2+(P+Q)\bar{\alpha}^4 \bar{\beta}^4}{2}\sum_{c=1}^{2C-1}\max_{m}\|\bM^c_m\|^2.
\end{align}
\end{lemma}

Having determined the Lipschitz constant $L$ for $\bg(\bz)$ and the variance bound $\sigma^2$ for $\hbg(\bz)$, the following proposition provides an upper bound on the total number of measurement samples running on all circuits.

\begin{proposition}\label{pro:S_all}
Given an arbitrary point $\bz^0$, run the EG iteration in~\eqref{eq:eg} for $T$ times with the step size $\mu^t_z=\frac{1}{2\sqrt{2}L}$. Accordingly, let $\{\bz^t\}_{t=1}^T$ and $\{\bbz^t\}_{t=0}^{T-1}$ denote the sequences of points generated by the EG iteration in~\eqref{eq:eg}. Choose a point $\hbz$ uniformly from $\{\bar{\bz}^t\}_{t=0}^{T-1}$. If the operator $\bg(\bz)$ satisfies Assumption~\ref{as:weak_mvi}, then after
\begin{equation}\label{eq:S_all}
((2P+1)(2C-1)+2Q+1)\left\lceil \frac{4,224 L^2 \sigma^2\|\bz^*-\bz^0\|_2^2}{\epsilon^4(1-4\sqrt{2}L\rho)^2}\right \rceil
\end{equation}
measurement samples across all iterations and for all circuits, it holds that
\[\mathbb{E}\!\left[\left\|\hbg(\hbz)\right\|_2\right]\leq \epsilon.\]
\end{proposition}

The proposition can be established by multiplying the number of primal/dual circuits $((2P+1)(2C-1)+2Q+1)$, the number of iterations in~\eqref{eq:num_iter}, and twice the number of samples in~\eqref{eq:S_one}. Here, we upper bound the factor $(8+\sqrt{2}L\rho)$ appearing in the numerator of~\eqref{eq:S_one} by $8.25$ for $\rho \in \left[0,\frac{1}{4\sqrt{2}L}\right)$.

Proposition~\ref{pro:S_all} establishes that the total number of measurement samples scales polynomially with the number of colors $C$, the number of PQC parameters $(P,Q)$, the Lipschitz constant $L$, and the variance bound $\sigma^2$. Numerical tests with the OPF in Section~\ref{subsec:xbm4ps} suggest that the number of colors scales as $C = \mathcal{O}((\log N)^3)$, a scaling we expect to extend to other problems with similar graph sparsity. Additionally, the parameter counts $P$ and $Q$ typically scale polynomial in  $\log N$~\cite{peruzzo2014}. The constants $L$ and $\sigma^2$ depend on problem-specific bounds $\bar{\alpha}$ and $\bar{\beta}$. If these bounds also scale polynomially in $\log N$, then the total number of measurement samples, and hence the overall complexity, becomes polynomial in  $\log N$, offering a potential advantage over the PD method running on a classical computer. The latter requires $\mathcal{O}(sN^2)$ operations, multiplying with the number of PD iterations in~\eqref{eq:pd1}. This favorable regime is achievable, for example, when optimizing over the probability simplex and the number of active constraints grows polynomially in $\log N$, or when finding ground states of Hamiltonians with bounded spectral norm. The next section models the OPF as a QCQP and expounds how the sparsity structure of the power network aligns with the measurement protocol presented in Section~\ref{sec:measure}.

%
\section{Optimal Power Flow as a QCQP}\label{sec:model}

The OPF is a large-scale, nonconvex optimization problem. Given forecasts of electric load demands for the next control period, the goal of the OPF is to schedule generators and flexible demand most economically while complying with engineering constraints imposed by the power transmission system and the generation units. This section states the OPF and poses it as a QCQP. 

Let us first review a power system model; see~\cite{redux} and references therein. An electric power system can be represented by an undirected connected graph $\mcG\coloneqq \{\mcN,\mcE\}$. The vertex set $\mcN\coloneqq\{1,\ldots,N\}$ comprises $N$ nodes. Nodes correspond to points where electric power is produced by generators or consumed by loads. The edge set $\mcE\coloneqq\{\ell=(n,m):n,m\in\mcN\}$ consists of $L_e\coloneqq|\mcE|$ transmission lines, each connecting two nodes. Because each node is connected only to a few other nodes, the number of lines, $L_e$, is typically a small multiple of $N$. 

The power system can be described by a vector $\bv\in\mathbb{C}^N$ of the complex AC voltages (voltage phasors) experienced at all nodes. Because many other quantities of interest can be expressed as functions of nodal voltages, the vector $\bv$ is typically selected as the state of the power system. For example, the vector of AC currents, $\bi$, injected into all nodes can be expressed as $\bi=\bY\bv$, where $\bY\in\mathbb{C}^{N\times N}$ is the node admittance matrix, which can be thought of as a complex-valued Laplacian matrix of the power system graph $\mcG$. Other grid quantities, such as the nodal power injections, line power flows, squared voltage magnitudes, and squared line magnitudes, can all be expressed as quadratic functions of~$\bv$, as $\bv^\dag\bM_m\bv$ for some Hermitian matrix $\bM_m$. If $v_n$ and $i_n$ are the AC voltage and current at node $n$, respectively, the complex power injected into the power network through node $n$ is defined as $s_n=v_ni_n^*=p_n+\iota q_n$, whose real part $p_n$ is termed active power, and imaginary part $q_n$ is reactive power. 

The cost function of the OPF is usually a linear function of some of the active power injections. To comply with generation capacities and load demands, the OPF imposes upper and lower limits on active and reactive power injections in the form of
\[\ubar{p}_n\leq p_n\leq \bar{p}_n\quad\text{and}\quad \ubar{q}_n\leq q_n\leq \bar{q}_n\quad \text{for all}~ n\in\mcN.\]
Nodal voltage magnitudes are constrained within upper and lower limits as
\[\ubar{v}_n^2\leq \left|v_n\right|^2\leq \bar{v}_n^2\quad \text{for all}~ n\in\mcN.\]
Line current magnitudes are upper-limited by line capacities as
\[ \left|i_{\ell}\right|^2\leq \bar{i}_{\ell}^2\quad \text{for all}~ \ell\in\mcE.\]

Since all quantities mentioned above are quadratic in~$\bv$, the OPF can be posed as a quadratically constrained quadratic program (QCQP) over $\bv$ as~\cite{lavaei2011,redux}:
where $b_m$, for all $m$, is a given real-valued parameter determined by lower/upper limits and each $\bM_m$ is a known Hermitian matrix. For example, if constraint $\ubar{p}_n\leq p_n$ is indexed as the $m$-th constraint, then $b_m=-\ubar{p}_n$. The OPF involves quadratic equality constraints, e.g., when $\ubar{p}_n=\bar{p}_n$ for a subset of nodes $n\in\mcN$. Each of these equality constraints corresponds to two inequality constraints in \eqref{eq:qc}. Overall, the number of OPF constraints can be several times the network size, such as $M\simeq 8N$. 

The precise form of $\bM_0$ and $\{(\bM_m,b_m)\}_{m=1}^M$ is delineated in Appendix~\ref{sec:app:opfmodel}. It is worth stressing that every matrix $\bM_m$ is highly sparse. This is because every node of a power system is directly connected to only a few other nodes. The number of nonzero entries per matrix $\bM_m$ is upper bounded by a small constant much smaller than $N$. The $(i,j)$-th off-diagonal entry across every matrix $\bM_m$ is nonzero only if nodes $i$ and $j$ are connected through a transmission line, that is, if $(i,j)\in\mcE$. Consequently, the sparsity pattern (location of nonzero entries) for each $\bM_m$ is determined by a small subset of $\mcE$. Therefore, the union of the sparsity patterns across all matrices $\bM_m$ coincides with the sparsity pattern of the nodal admittance matrix $\bY$. The sparsity of the matrices ${\bM_m}$ combined with node permutation allows the number of colors $C$ to scale favorably as $\mathcal{O}((\log N)^3)$, as demonstrated in Section~\ref{subsec:implement}. This efficient coloring enables the measurement protocol of Section~\ref{sec:measure}.

The favorable scaling of the number of colors $C$ does not, by itself, guarantee a practical quantum advantage for the OPF. This is because the overall sample complexity in Proposition~\ref{pro:S_all} also depends on the Lipschitz constant $L$ and the variance bound $\sigma^2$, both of which scale unfavorably for the OPF. Specifically, the magnitudes of the entries in $\bv$ in~\eqref{eq:qc} are constrained to lie in the range of $[0.95,1.05]$, so that $\bar{\alpha} = 1.05\sqrt{N}$. Furthermore, because OPF imposes equality constraints on nodes connected to loads, each such equality constraint is converted into two inequality constraints in~\eqref{eq:qc}. Since these inequalities should be active at optimality, the complementary slackness of the Karush–Kuhn–Tucker conditions written for~\eqref{eq:qc} implies that the associated dual variables are generally nonzero~\cite[Ch.~5]{BoVa04}. Because the number of load nodes in power systems is typically proportional to $N$, it follows that $\bar{\beta}$ may scale with $\sqrt{N}$ too. Overall, in practice, parameters $L$ in~\eqref{eq:lipschitz} and $\sigma^2$ in~\eqref{eq:sigma:1} could grow as $N^2$ and $N^4$, respectively. This indicates that the total number of measurement samples in~\eqref{eq:S_all} could scale with $N^8$. Unless quantum measurements can be recorded really fast, such a large number of measurement samples ultimately presents a challenge to attaining quantum advantage. 
This unfavorable scaling of quantum measurements is not limited to the OPF, but applies to a broader class of optimization problems when solved quantumly. We next elaborate on this by discussing the measurement sample complexity for a variational quantum algorithm for a prototypical quadratic optimization problem.

\subsection{Discussion of measurement sample complexity}
Consider the unconstrained quadratic optimization problem
\begin{equation}\label{eq:quad}
    \min_{\bx:\left\|\bx\right\|_2=1} \bx^{\dag} \bH\bx,
\end{equation}
where $\bH \in \mathbb{C}^{N \times N}$ is a Hermitian matrix. Analogous to the VQE idea applied to~\eqref{eq:vqe}, the corresponding parameterized problem of~\eqref{eq:quad} is
\begin{equation}\label{eq:F}
    \min_{\btheta\in [0,2\pi]^P} F(\btheta),
\end{equation}
where $F(\btheta):=\braket{\bpsi(\btheta)|\bH|\bpsi(\btheta)}$. Since $F(\btheta)$ is a trigonometric polynomial in $\btheta$, the gradient of $F(\btheta)$ is Lipschitz continuous. Moreover, the parameter‑shift rule in~\eqref{eq:psr} provides unbiased estimates for $\nabla_{\btheta} F(\btheta)$~\cite{sweke2020}. Standard convergence results suggest that the number of stochastic gradient descent (SGD) iterations needed to reach an $\epsilon$-stationary point of \eqref{eq:F} is~\cite{ghadimistochastic}:
\begin{equation}\label{eq:T}
    T=\mcO\!\left(\frac{L\sigma^2}{\epsilon^2}\right),
\end{equation}
where $L$ is the Lipschitz constant of the objective gradient, and $\sigma^2$ is an upper bound on the variance of the stochastic gradient estimator for all iterations.

By arguments analogous to those in the proofs of Lemmas~\ref{le:lipschitz} and~\ref{le:var}, the two constants can be estimated as 
\begin{subequations}\label{eq:unconstrained}
\begin{align}
    L &= \mcO\!\left(P\|\bH\|\right), \label{eq:L:unconstrained}\\
    \sigma^2 &= \mcO\!\left(\frac{P\max_{\btheta}\Var(\htF(\btheta))}{S}\right), \label{eq:sigma:unconstrained}
\end{align}    
\end{subequations}
where $\htF(\btheta)$ is the sample-based estimator of $F(\btheta)$. 
By Popoviciu's inequality, the variance $\max_{\btheta}\Var(\htF(\btheta))$ in~\eqref{eq:sigma:unconstrained} can be bounded by 
\begin{equation}\label{eq:bound:var}   \max_{\btheta}\Var(\htF(\btheta))\leq \frac{(\lambda_{\max}-\lambda_{\min})^2}{4},
\end{equation}
where $\lambda_{\max}$ and $\lambda_{\min}$ denote the largest and smallest eigenvalues of $\bH$, respectively. The equality is attained when $\ket{\bpsi(\btheta)}$ is an equal superposition of the eigenvectors corresponding to $\lambda_{\max}$ and $\lambda_{\min}$.

Since each gradient estimate requires $2P$ shifted circuit evaluations, the total number of quantum measurements across $T$ iterations is $2PST$. Substituting \eqref{eq:unconstrained} and \eqref{eq:bound:var} into \eqref{eq:T} yields the total quantum measurement sample complexity
\begin{equation}\label{eq:S_total}
    2PST=\mcO\!\left(\frac{P^3\|\bH\|(\lambda_{\max}-\lambda_{\min})^2}{\epsilon^2}\right).
\end{equation}

In the special case where $\bH$ is positive definite, its eigenvalues coincide with its singular values. Let $\sigma_{\min}$ be the minimum singular value of $\bH$, and let $\kappa=\|\bH\|/\sigma_{\min}$ be the condition number of $\bH$. Then~\eqref{eq:S_total} can be compactly expressed as
\begin{align}\label{eq:S_total:2}
    2PST=\mcO\!\left(\frac{P^3(\kappa\sigma_{\min})^3}{\epsilon^2}\right).
\end{align}
Thus, if $(\kappa\sigma_{\min})^3$ grows with $N$, the total number of quantum measurement samples also grows with $N$. 

So far, in~\eqref{eq:quad} we have considered $\bx$ with unit Euclidean norm. If $\bx$ is allowed to have an arbitrary norm that is upper bounded by $\bar{\alpha}$ as in~\eqref{eq:modelv}, the previous analysis reduces to the unit‑norm case, yet with $\bH$ scaled by $\bar{\alpha}^2$. As a result, $\sigma_{\min}$ is multiplied by $\bar{\alpha}^2$, whilst $\kappa$ remains unchanged. Accordingly, the measurement sample complexity is multiplied by $\bar{\alpha}^6$. 

We now provide two examples where the total measurement sample complexity scales favorably or unfavorably with $N$. We begin with the favorable case. Consider the problem of minimum-variance portfolio rebalancing~\cite{won2013condition}. Let $\bx\in\mathbb{R}^N$ be the vector of portfolio weights of $N$ assets. Let $\bSigma\in\mathbb{R}^{N\times N}$ denote the covariance matrix of the $N$ asset returns. The objective is
\begin{equation}\label{eq:gmvp}
    \min_{\bx:\left\|\bx\right\|_2=1} \bx^\top\bSigma\,\bx.
\end{equation}
The problem in \eqref{eq:gmvp} complies with the form in \eqref{eq:quad} with $\bH=\bSigma\succeq 0$. In practice, the sample covariance matrix can be ill‑conditioned. However, regularized covariance estimators, such as shrinkage methods in~\cite{won2013condition}, lift small eigenvalues and control $\kappa$ independent of $N$. Moreover, in this regularized regime, the smallest eigenvalue $\lambda_{\min}$ can be treated as a constant with respect to $N$~\cite{won2013condition}. Subsequently, the total measurement sample complexity is $\mcO(P^3/\epsilon^2)$, which does not scale with $N$.

In the unfavorable case, let us simplify the OPF setting by ignoring constraints and working with the DC power flow approximation. Specifically, consider $\bH = \bL$ and $\bx=\bv$, where $\bL$ is the DC bus admittance matrix. Matrix $\bL$ is a real symmetric, positive semidefinite graph Laplacian of the power network. It has been numerically observed that $\kappa(\bL)$ scales with $N$~\cite{demystifying}. In addition, as discussed earlier, the scaling factor $\bar{\alpha}$ scales as $\sqrt{N}$. Consequently, both $\kappa^3$ and $\bar{\alpha}^6$ grow as $N^3$. Therefore, the total measurement sample complexity grows as $N^6$. If we also account for the constraints, the complexity increases further.

Despite its high sample complexity, the doubly variational quantum model may still offer practical value for the OPF. The proposed framework can be readily extended to a data‑based setting to solve parameterized optimization problems. To be more specific, our proposed formulation handled a single instance of the OPF, with the instance specifications (corresponding to $b_m$ in~\eqref{eq:qc}) appearing only in the observable $\bS$ in~\eqref{eq:term3}. If, instead, these specifications are embedded into the primal and dual PQCs, the PQC parameters can be trained across multiple instances to learn the OPF mapping. In operation, the trained PQCs can predict the OPF solutions without repeatedly solving the OPF from scratch. This idea has been explored in learning the AC power flow solutions using data-embedded PQCs~\cite{SGC25}.
\color{black}

\section{Numerical Tests}\label{sec:tests}
The performance of the proposed doubly variational quantum QCQP was validated for the OPF problem (hereafter referred to as Q-OPF) using the IEEE 57-node power system benchmark~\cite{pglib}. For this benchmark, we have $N=57$ primal variables and $M=422$ dual variables. Therefore, the number of qubits needed to capture $\bv(\btheta)$ and $\blambda(\bphi)$ is 6 and 9, respectively. To simplify the implementation, nodes hosting generators were assumed to have no load, i.e., $p^d_n=q^d_n=0$ for all $n\in\mcN_g$. Reactive power demands were set so that load nodes have $q^d_n=0.33\times p^d_n$ for all $n\in\mcN_l$. We generated 15 different instances of the OPF by scaling original load demands $\{(p^d_n,q^d_n)\}_{n \in \mcN_{\ell}}$ by factors drawn uniformly distributed within $[0.90,1.05]$, independently per node and problem instance. Each of the 15 OPF instances was solved to global optimality using MATPOWER, a power system toolbox in MATLAB~\cite{matpower}.

\emph{Selecting PQC architecture.} To select the architecture of PQC$_p$, we numerically tested different architectures. Each candidate architecture was validated over all 15 problem instances. For each problem instance $k$, we found the OPF solution $\bv_k^*$ using MATPOWER, and then solved the variational problem
\begin{equation*}
\min_{\btheta_k}\left(1-\frac{\braket{\bpsi(\btheta_k)|\bv_k^*}}{\|\bv_k^*\|_2}\right)
\end{equation*}
with $\ket{\bpsi(\btheta_k)}$ generated by the specified PQC$_p$ architecture. Assuming the previous cost is optimized to global optimality, it measures the alignment between the generated quantum state and the actual OPF solution. The cost is inspired by~\cite{BravoPrieto2023variationalquantum} and is lower-bounded by zero. We tested architectures with different gate types and number of layers $L$, as detailed in Appendix~\ref{sec:app:vqc}. We selected the PQC that attained the smallest average cost across all 15 OPF problem instances. A similar procedure was followed to determine the PQC$_d$ architecture. The optimal PQC$_p$ architecture consists of 10 layers, with each layer comprising the sequence $\bR_y(\theta)-\text{CX}-\bR_z(\theta')-\text{CX}$ gates, resulting in $P=2\times 10\times 6=120$ parameters. The optimal PQC$_d$ architecture is formed by 35 layers of $\bR_y(\phi)-\text{CX}$ gates, yielding $Q=1\times 35\times 9=315$ parameters. The average fitting errors obtained were $2\times 10^{-4}$ and $2\times 10^{-5}$ for PQC$_p$ and PQC$_d$, respectively. This demonstrates that the selected PQC pairs are sufficiently expressive. 

\emph{Algorithm details.} PQC parameters $\btheta$ and $\bphi$ were initialized uniformly at random within the range $[0,2\pi]$. Because voltages are normally around 1 per unit, the scaling variable $\alpha$ was initialized to $\sqrt N$; the variable $\beta$ was initialized to $2\left|\mcN_{\ell}\right|$. Step sizes were set according to the exponentially decaying rule as $\mu^t_{\theta}=0.015\times 0.99985^t$, $\mu^t_{\phi}=0.01\times 0.99985^t$, and $\mu^t_{\alpha}=\mu^t_{\beta}=10^{-5}\times 0.999^t$. The bound $\bar{\alpha}$ was set to $1.05\sqrt{N}$, whilst $\bar{\beta}$ was set as 500. Both PD and EG iterations were terminated when both conditions $\left\|\btheta^t-\btheta^{t-1}\right\|_2\leq 10^{-6}$ and $\left\|\bphi^t-\bphi^{t-1}\right\|_2\leq 10^{-6}$ were satisfied. Simulation scripts were written in Python and run on Pennylane's exact simulator~\cite{Pennylane}.

\begin{table}[]
\centering

\begin{tabular}{l|c|c|c|c|c}
 \hline \hline
\textbf{Model} & $\bx_g$ [\%] & $\blambda$ [\%] & \textbf{(a)} & \textbf{(b)} & \textbf{(c)} \\ \hline\hline
   \text{QCQP-PD}   & 14.95 & 19.30        &  25.87            &    13.73          &  0.49 
   \\ \hline
   \text{QCQP-EG}   & 13.67 & 19.20        &  24.53            &    12.29          &  0.39
  \\ \hline \hline 
    $\text{QCQP}_{\theta}\text{-PD}$    & 11.82 &  12.87      &  11.73            & 13.03             & 0.23            \\ \hline
  $\text{QCQP}_{\theta}\text{-EG}$    & 7.62 &  12.17      &  11.53            & 11.86             & 0.21 \\\hline\hline
\end{tabular}
\caption{Average relative errors and constraint violation statistics across problem instances: (a) Number of constraint violations per instance; (b) maximum constraint violation [\%]; (c) average constraint violation [\%].}
\label{tab:error}
\end{table}

We first explored the performance of Q-OPF in finding AC OPF solutions. Rather than evaluating errors on OPF solutions in terms of voltage phasors $\bv$, we focused on finding the active power injection and voltage magnitude at all generator nodes. This is because power system operators are primarily interested in generator setpoints. Define the vector of generator setpoints $\bx_g=[\bp_g;\bv_g]^\top$ with $\bp_g$ and $\bv_g$ collecting generator active power injections and voltage magnitudes, respectively. Generator active power injections can be computed from \eqref{eq:qccon1} given $\bv$ and the corresponding load demand $p_n^d$. In terms of dual variables, we focus on Lagrange multipliers $\lambda_m$ associated with power balance and line limit constraints because these are used in electricity markets to compute the so-called \emph{locational marginal prices}.

\begin{figure}[t]
\centering
\includegraphics[width=0.85
\linewidth]{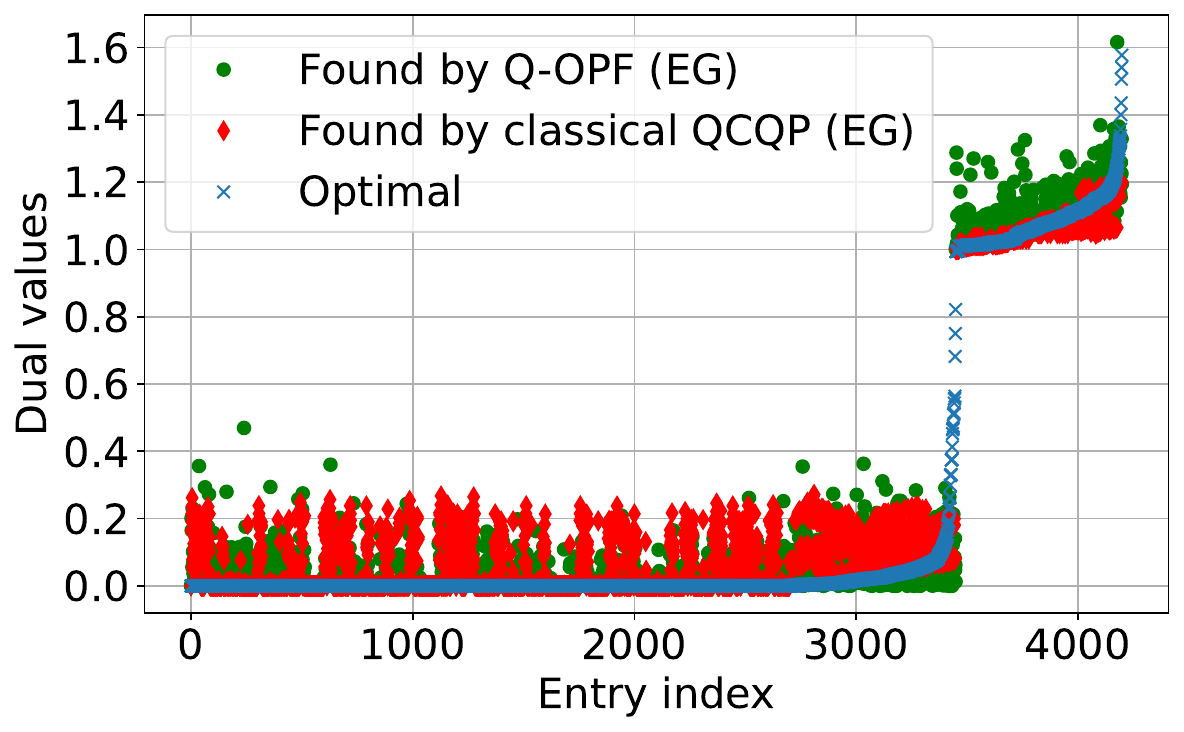}
\caption{Comparing dual entries obtained by the Q-OPF and the classical QCQP, both solved by the EG iteration, to the optimal values found by MATPOWER. Dual vectors of 15 instances were concatenated and sorted in increasing order. Entries smaller than $10^{-6}$ were assigned to 0.}
\label{fig:dual}
\end{figure}

We compared solving \eqref{eq:qc} and \eqref{eq:qcF} in terms of the relative errors $\left\|\bx_g-\bx_g^*\right\|_2/\left\|\bx_g^*\right\|_2$ and $\left\|\blambda-\blambda^*\right\|_2/\left\|\blambda^*\right\|_2$ of the found generator setpoints $\bx$ and dual variables $\blambda$ from their global optimal solutions found using MATPOWER. The entries of the dual vector $\blambda$ correspond to the Lagrangian multipliers associated with 280 power balance and line limit constraints. The non-convex \eqref{eq:qc} was solved classically using the PD iteration in~\eqref{eq:pd1} and the EG iteration in~\eqref{eq:eg} over voltage phasors $\bv$ and dual variables $\blambda$. Voltages were initialized at the flat voltage profile as $\bv^0=\bone$, while the entries of $\blambda$ were initialized independently at random based on the standard normal distribution and scaled by $2\left|\mcN_{\ell}\right|$. The related step sizes were set as $\mu_v=\mu_{\lambda}=10^{-3} \times 0.9999^t$, and PD iterates were terminated when $\left\|\bv^t-\bv^{t-1}\right\|_2\leq 10^{-6}$ and $\left\|\blambda^t-\blambda^{t-1}\right\|_2\leq 10^{-6}$ were met. 

Table~\ref{tab:error} shows the relative errors for generator setpoints and dual variables attained by solving~\eqref{eq:qcF} and~\eqref{eq:qc} via the PD and EG iterations. Among the four models, solving~\eqref{eq:qcF} using the EG iteration achieved the smallest relative errors on $\bx_g$ and $\blambda$. Comparing algorithms when solving either~\eqref{eq:qcF} or~\eqref{eq:qc}, the EG method consistently outperforms the PD method. Similarly, fixing the algorithm to be either EG or PD, solving~\eqref{eq:qcF} using a hybrid quantum-classical computer offers smaller errors compared to solving~\eqref{eq:qc} using the PD method.

Given that the EG iteration outperforms the PD iteration, Fig.~\ref{fig:dual} demonstrates the dual variables obtained by solving~\eqref{eq:qcF} and~\eqref{eq:qc} by the EG method. The Q-OPF approach can approximate dual variables reasonably well. This test corroborates that Q-OPF is capable of finding near-optimal AC OPF solutions.

\emph{Feasibility.} Does the power system safely operate under the generator setpoints computed by Q-OPFs? To answer this question, for each instance, given the input load demands and the setpoints $\bx_g$ obtained by Q-OPF, we found an AC power flow solution via MATPOWER. The obtained voltages were used to evaluate all 222 inequality constraints in~\eqref{eq:qc_detail}. Violations in line current and generator constraints were normalized by their maximum limits. We employed three constraint violation metrics: \emph{a)} the average number of constraint violations surpassing a normalized magnitude of $10^{-6}$ across OPF instances; \emph{b)} the maximum constraint violation incurred across instances; \emph{c)} the violations averaged over all constraints and instances. As shown in Table~\ref{tab:error}, solving~\eqref{eq:qcF} using the EG iteration attains the best constraint violation metrics among the tested models.

Figure~\ref{fig:Lagrangian} shows the relative errors in terms of the Lagrangian function $\left|\mcL-P^*_{\theta}\right|/P^*_{\theta}$ for the last EG iteration. The relative errors of the Q-OPF are consistently below 1.5\% and smaller than those of QCQP in 14 out of 15 instances. This test indicates that under the Q-OPF's generator setpoints, all physical constraints of the network are nearly satisfied.

\begin{figure}[t]
\centering
\includegraphics[width=0.85
\linewidth]{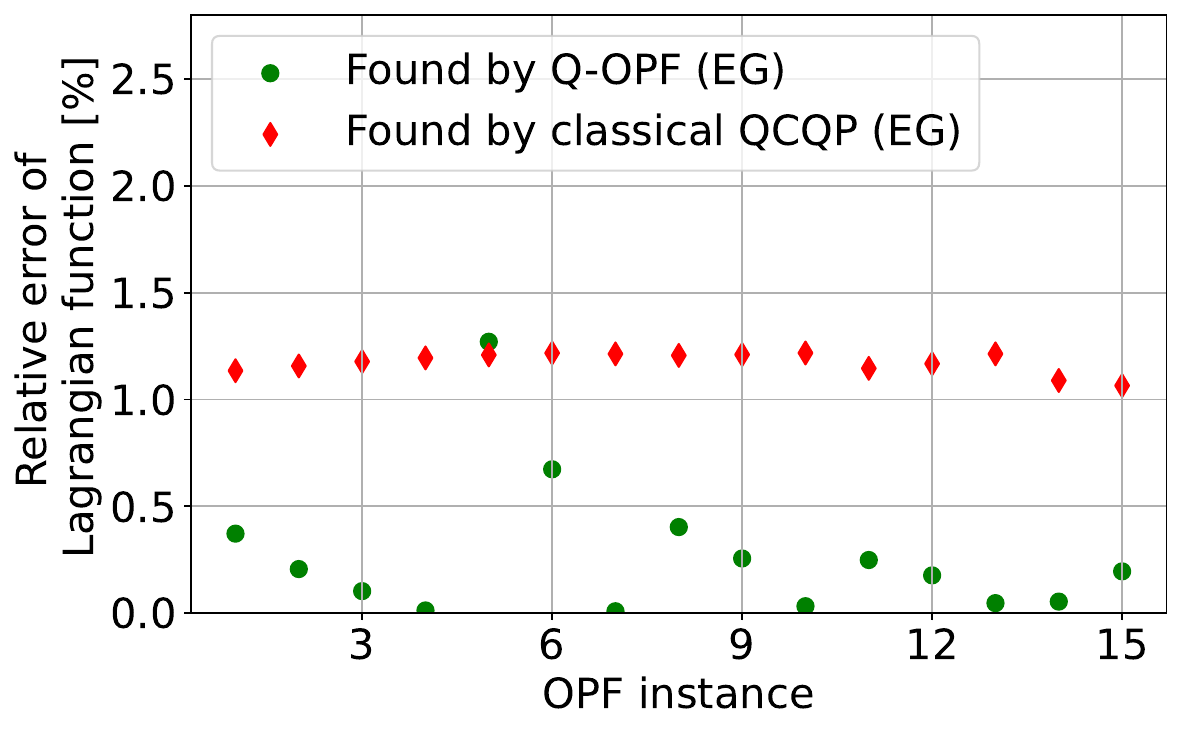}
\caption{Comparing the relative error of the Lagrangian function for Q-OPF and the classical non-convex QCQP solved by the EG iteration across 15 OPF instances. The globally optimal value of the Lagrangian was computed using MATPOWER.}
\label{fig:Lagrangian}
\end{figure}

\section{Conclusions}\label{sec:conclude}

This work proposes a doubly variational quantum approach to solve QCQP/SDPs, where the observables are induced by underlying sparse graphs. Primal/dual variables are modeled by scaling the state and the PMF corresponding to two PQCs. These two PQCs are trained using exclusively the gradient estimates of the Lagrangian function. The training procedure aims to seek an approximately stationary point of the Lagrangian function. By adopting the XBM method and leveraging graph sparsity, the related observables are judiciously permuted into a banded form, enabling efficient evaluation of the Lagrangian function. {In particular, the sample complexity analysis in Section~\ref{sec:analysis} characterizes regimes under which the proposed doubly variational quantum model may achieve quantum advantage.} The proposed algorithmic framework is applied to the OPF problem. Numerical tests on the IEEE 57-node benchmark power system using PennyLane's quantum simulator have demonstrated that: \emph{i)} the EG iteration performs markedly better than the PD iteration; \emph{ii)} solving~\eqref{eq:qcF} on PQCs offers smaller errors in terms of the generator setpoint and dual variables, as well as smaller constraint violations compared to solving~\eqref{eq:qc} on a classical computer using PD; and \emph{iii)} the Q-OPF found generator setpoints and dual variables with small errors, and under the generator setpoints obtained by the Q-OPF, the network constraints are nearly satisfied.

The proposed doubly variational quantum paradigm sets a solid foundation for exploring exciting directions: \emph{d1)} Thus far, the PQC parameters $(\btheta,\bphi)$ have been optimized to solve one problem instance, and parameters $b_m$ have solely appeared on the diagonal of the observable $\bS$. Nevertheless, vector $\bb$ can be encoded as parameters of another unitary $\bT(\bb)$, under which the corresponding primal and dual states are  $\ket{\bpsi(\btheta,\bb)}=\bT(\bb)\bU(\btheta)\ket{\bzero}$ and $\ket{\bxi(\bphi,\bb)}=\bT(\bb)\bV(\bphi)\ket{\bzero}$. Measuring the states $\ket{\bpsi(\btheta,\bb)}$ and $\ket{\bxi(\bphi,\bb)}$ and utilizing the proposed doubly variational algorithm for the Lagrangian in~\eqref{eq:Lagrangian1} evaluated across instances can solve multiple problem instances; \emph{d2)} Prudently construct an alternative objective function for~\eqref{eq:qc} to escape the problem of measurement samples scaling with $\bar{\alpha}$ and $\bar{\beta}$; \emph{d3)} Designing network-informed PQCs and/or testing the Q-OPF using dynamic PQCs and quantum recurrent embedding neural networks will be imperative for dealing with large-scale networks; and \emph{d4)} The Q-OPF can be a building block to investigate other computational tasks, such as multi-period OPF, which entails the time-coupled ramp constraints on generators.

\section{Data Availability}
{The code used to generate the results in this study is available at this repository: \href{https://github.com/LeQuantum/Variational-quantum-algorithms-for-solving-conic-programs-over-sparse-graphs}{\texttt{VQA-for-conic-programs-over-sparse-graphs}}}.

\begin{acknowledgments}
VK and TVL acknowledge support from the US National Science Foundation under grant no.~2412947, and the US Office of Naval Research under grant N000142412614. MMW acknowledges support from the National Science Foundation under grant no.~2329662 and from the Cornell School of Electrical and Computer Engineering.
\end{acknowledgments}

\bibliography{Q_OPF}
\appendix
\section{Appendix}

\subsection{Power system modeling}\label{sec:app:opfmodel}
A transmission line connecting node $n$ to node $m$ is represented by its series conductance $G_{nm}$ and series susceptance $B_{nm}$. Given $G_{nm}$ and $B_{nm}$, power injections at node $n$ are computed through nodal voltages based upon the following quadratic power flow equations:
\[p_n = \sum_{m=1}^{N} v_n^r(v_m^rG_{nm}-v_m^iB_{nm}
)+v_n^i(v_m^iG_{nm}+v_m^rB_{nm}),\]
\[q_n = \sum_{m=1}^{N} v_n^i(v_m^rG_{nm}-v_m^i B_{nm}
)-v_n^r(v_m^iG_{nm}+v_m^r B_{nm}),\]
where $(v_n^r,v_n^i)$ and $(v_m^r,v_m^i)$ are the real and imaginary components of the voltages at nodes $n$ and $m$, respectively.
Define the $N \times N$ node admittance matrix as $\bY\coloneqq \bG+\iota\bB$, where $\bG$ collects $G_{nm}$ and $\bB$ collects $B_{nm}$ of the network lines.  Then the power flow equations can be compactly represented as
\begin{align}
    \label{opf1}
    p_n & = \bv^\dag \bM_{p_n} \bv ,\\
    \label{opf2}
    q_n & = \bv^\dag \bM_{q_n} \bv,
\end{align}
where $\bM_{p_n}$ and $\bM_{q_n}$ are $N \times N$ Hermitian matrices taken the forms as 
\[\bM_{p_n} \coloneqq  \frac{1}{2} (\bY^\dag \be_n \be_n^\top+\be_n \be_n^\top \bY),\]
\[\bM_{q_n} \coloneqq  \frac{1}{2\iota} (\bY^\dag \be_n \be_n^\top-\be_n \be_n^\top \bY).\]
Similarly, the squared voltage magnitude at node $n$ is represented as
\begin{equation}
    (v_n^r)^2+(v_n^i)^2=\bv^\dag \bM_{v_n} \bv
\end{equation}
where $\bM_{v_n}\coloneqq \be_n\be_n^\top$.
Given a line $(n,m) \in \mcE$ with the series admittance $Y_{nm}$, the complex current flowing on from node $n$ to node $m$ is provided as $i_{nm}=Y_{nm}(v_n-v_m)$. This poses $|i_{nm}|=\bv^\dag \bM_{i_{nm}} \bv$
where $ \bM_{i_{nm}}\coloneqq \left |Y_{nm}\right |(\be_n-\be_m)(\be_n-\be_m)^\top$.

The power injection $(p_n,q_n)$ consists of a dispatchable constituent $(p^g_n,q^g_n)$ and an inflexible constituent $(p^d_n,q^d_n)$ such that $p_n=p^g_n-p^d_n$ and $q_n=q^g_n-q^d_n$. The former represents the power dispatch of a generator or a flexible load placed at node $n$, and the latter expresses the inelastic load supplied at node $n$. Let nodes with dispatchable power injections belong to a subset $\mcN_g \subseteq \mcN$. The remaining nodes hosted inelastic loads form the subset $\mcN_{\ell}=\mcN  \backslash \mcN_g$. Node $n=1$ is referred to as the angle reference node, under which it follows $\bv^\dag \be_1 \be_1^\top \bv=1$, where $\be_1$ is the first column of the identity. To simplify the exposition, each node in $\mcN_g$ is assumed to host only one dispatchable unit.
Given the inflexible loads $\{p^d_n,q^d_n\}_{n \in \mcN}$, the OPF's objective is to minimize the dispatch cost of generators and flexible load while respecting the resource and network limits. The OPF is formulated as the ensuing QCQP over nodal voltages and power of generators:
\begin{subequations}\label{eq:qc_detail}
\begin{align} 
    \min~&~\sum_{n \in \mcN_g} c_n p^g_n \\
    \mathrm{over}~&~\bv\in \mathbb{C}^N, \{p^g_n,q^g_n\}_{n \in \mcN_g} \nonumber \\
    \mathrm{s.to}~&~\bv^\dag \bM_{p_n} \bv=p^g_n-p^d_n 
        & \forall n \in \mcN_g \label{eq:qccon1} \\
    ~&~\bv^\dag \bM_{q_n} \bv=q^g_n-q^d_n 
        & \forall n \in \mcN_g \label{eq:qccon2}\\
    ~&~\bv^\dag \bM_{p_n} \bv=-p^d_n 
        & \forall n \in \mcN_{\ell} \label{eq:qccon3} \\
    ~&~\bv^\dag \bM_{q_n} \bv=-q^d_n 
        & \forall n \in \mcN_{\ell} \label{eq:qccon4} \\
    ~&~ \ubar{p}^g_n \leq \bv^\dag \bM_{p_n}\bv+p^d_n \leq \bar{p}^g_n 
        & \forall n \in \mcN_g \label{eq:qccon5}\\
    ~&~ \ubar{q}^g_n \leq \bv^\dag \bM_{q_n}\bv+q^d_n \leq \bar{q}^g_n 
        & \forall n \in \mcN_g \label{eq:qccon6}\\
    ~&~ \ubar{v}_n \leq \bv^\dag \bM_{v_n}\bv \leq \bar{v}_n 
        & \forall n \in \mcN \label{eq:qccon7}\\
    ~&~\bv^\dag \be_{1}\be_{1}^\top \bv=1 \label{eq:qccon8}\\
    ~&~\bv^\dag \bM_{i_{nm}}\bv \leq \bar{i}_{nm} 
        & \forall (n,m) \in \mcE \label{eq:qccon9}
\end{align}
\end{subequations}
Constraints~\eqref{eq:qccon1}--\eqref{eq:qccon4} capture the power flow equations at generator and load nodes. Constraints~\eqref{eq:qccon5}--\eqref{eq:qccon6} confine limits of generators and flexible loads. Constraint~\eqref{eq:qccon7} limits the squared voltage magnitudes within the specified range, and constraint~\eqref{eq:qccon8} represents the squared magnitude of the angle reference node. Constraint~\eqref{eq:qccon9} enforces the line current magnitude within the line ratings.
By substituting $\{p^g_n,q^g_n\}$ in~\eqref{eq:qccon1}--\eqref{eq:qccon2} into the objective of~\eqref{eq:qc_detail}, the OPF model is represented as a QCQP in terms of the nodal voltage vector $\bv$.

\subsection{Proof of Lemma~\ref{le:term2}}\label{sec:app:le1}
The claim follows readily by substituting $\bM$ into the observable as
\begin{align*}
\braket{\bxi,\bpsi|\bM|\bxi,\bpsi}&=(\bxi\otimes \bpsi)^\dag\left(\sum_{m=1}^M \be_m \be_m^\top \otimes \bM_m\right) (\bxi\otimes \bpsi)\\
&=\sum_{m=1}^M(\bxi\otimes \bpsi)^\dag\left( \be_m \be_m^\top \otimes \bM_m\right) (\bxi\otimes \bpsi)\\
&=\sum_{m=1}^M\left( \bxi^\dag\be_m \be_m^\top\bxi\right) \otimes \left(\bpsi^\dag\bM_m\bpsi\right)\\
&=\sum_{m=1}^M|\xi_m(\bphi)|^2F_m(\btheta).\\
\end{align*}

\subsection{Proof of Lemma~\ref{le:lipschitz}}\label{sec:app:le3}

Quantum observables are trigonometric functions of $(\btheta,\bphi)$. Moreover, the Lagrangian function $\mcL(\bz)$ is quadratic in $(\alpha,\beta)$. Therefore, the operator $\bg(\bz)$ is continuously differentiable with respect to $\bz$. Let $\bJ(\bz)\coloneqq \nabla_{\bz}\bg(\bz)$ denote the Jacobian matrix of $\bg(\bz)$. If the spectral norm of $\bJ(\bz)$ is upper bounded by a constant $L$, then $\bg(\bz)$ is $L$-Lipschitz continuous~\cite[Lemma~2.6]{garrigos2023handbook}. Based on matrix norm inequalities, the spectral norm of $\bJ(\bz)$ can be bounded as:
\begin{align}\label{eq:matrix_norm}
\left\|\bJ(\bz)\right\| &\leq \sqrt{\left\|\bJ(\bz)\right\|_1 \left\|\bJ(\bz)\right\|_{\infty}}
\end{align}
where $\left\|\bJ(\bz)\right\|_1$ is the maximum absolute column sum of $\bJ(\bz)$, and $\left\|\bJ(\bz)\right\|_{\infty}$ is the maximum absolute row sum. Because $\mcL(\bz)$ has continuous second-order partial derivatives, Schwarz's theorem predicates that $\frac{\partial^2\mcL}{{\partial z_i\partial z_j}}=\frac{\partial^2\mcL}{{\partial z_j\partial z_i}}$, so that $\left\|\bJ(\bz)\right\|_1=\left\|\bJ(\bz)\right\|_{\infty}$ and \eqref{eq:matrix_norm} simplifies as
\begin{align}\label{eq:lipschitz:bound}
\left\|\bJ(\bz)\right\| &\leq \left\|\bJ(\bz)\right\|_{\infty}=\max_{j\in\{1,\ldots,P+Q+2\}}\left\|\be_j^{\top}\bJ(\bz)\right\|_1 
\end{align}
where $\be_j$ is the $j$-th column of the identity.

Consider first the first $P$ rows of $\bJ(\bz)$ corresponding to the partial Hessian matrix $\nabla_{\btheta,\bz}^2\mcL(\bz)$. More specifically, the absolute row sum for row $j\in\{1,\ldots,P\}$ is
\begin{multline}\label{eq:hessian:theta}
\left\|\be_j^{\top}\bJ(\bz)\right\|_1 = \left(\sum_{p=1}^P \left|\frac{\partial^2\mcL(\bz)}{\partial\theta_j \partial\theta_{p}}\right|\right)
+ \left|\frac{\partial^2\mcL(\bz)}{\partial\theta_j\partial\alpha}\right|
 \\
+ \left(\sum_{q=1}^Q \left|\frac{\partial^2\mcL(\bz)}{\partial\theta_j\partial\phi_q}\right|\right)+ \left|\frac{\partial^2\mcL(\bz)}{\partial\theta_j\partial\beta}\right|.
\end{multline}
Applying the triangle inequality on each summand of the first term in~\eqref{eq:hessian:theta} yields
\[\left|\frac{\partial^2\mcL(\bz)}{\partial\theta_j\partial\theta_{p}} \right| \leq \alpha^2 \left|\frac{\partial^2F_0(\btheta)}{\partial \theta_i\partial \theta_{p}}\right|+\alpha^2\beta^2\sum_{m=1}^M |\xi_m(\bphi)|^2\left|\frac{\partial^2F_m(\btheta)}{\partial \theta_i\partial \theta_{p}}\right| \]
The second-order partial derivatives of quantum expectations can be computed according to the PSR as~\cite{mari2021}:
\begin{align*}
\left|\frac{\partial^2F_m(\btheta)}{\partial\theta_j\partial\theta_{p}}\right|
&=\frac{1}{4}\left|F_m\left(\btheta^+_j+\frac{\pi}{2} \be_{p}\right)-F_m\left(\btheta^+_j-\frac{\pi}{2} \be_{p}\right)\right.\nonumber\\
&\quad\left.-F_m\left(\btheta^-_j+\frac{\pi}{2} \be_{p}\right)+F_m\left(\btheta^-_j-\frac{\pi}{2} \be_{p}\right)\right| \nonumber\\
&\leq \frac{1}{4}\left|F_m\left(\btheta^+_j+\frac{\pi}{2} \be_{p}\right)\right|+\frac{1}{4}\left|F_m\left(\btheta^+_j-\frac{\pi}{2} \be_{p}\right)\right|\nonumber\\
&+\frac{1}{4}\left|F_m\left(\btheta^-_j+\frac{\pi}{2} \be_{p}\right)\right|+\frac{1}{4}\left|F_m\left(\btheta^-_j-\frac{\pi}{2} \be_{p}\right)\right| \nonumber\\
&\leq \left\|\bM_m\right\|
\end{align*}
for all $m$, $p$, and $j$.

Based on the latter bounds, we get that:
\begin{align}\label{eq:theta:theta}
&\sum_{p=1}^P\left|\frac{\partial^2\mcL(\bz)}{\partial\theta_j\partial\theta_{p}} \right| \leq P\alpha^2 \left(\left\|\bM_0\right\|+\beta^2\sum_{m=1}^M|\xi_m(\bphi)|^2 \left\|\bM_m\right\|\right)\nonumber\\
&\quad\quad\leq P\alpha^2 \left(\left\|\bM_0\right\|+\beta^2\braket{\bxi(\bphi)|\bxi(\bphi)}\max_{m}\left\|\bM_m\right\|\right)\nonumber\\
&\quad\quad=P\alpha^2 \left(\left\|\bM_0\right\|+\beta^2\max_{m}\left\|\bM_m\right\|\right).
\end{align}

The second term of~\eqref{eq:hessian:theta} can be upper bounded similarly using the PSR for first-order partial derivatives:
\begin{align}\label{eq:theta:alpha0}
\left|\frac{\partial^2\mcL(\bz)}{\partial\theta_j\partial \alpha}\right| 
&\leq 2\alpha\left|\frac{\partial F_0(\btheta)}{\partial\theta_j}\right|+2\alpha\beta^2 \sum_{m=1}^M|\xi_m(\bphi)|^2\left|\frac{\partial F_m(\btheta)}{\partial \theta_j}\right|\nonumber\\
&\leq 2\alpha\left\|\bM_0\right\|+2\alpha\beta^2\sum_{m=1}^M|\xi_m(\bphi)|^2\left\|\bM_m\right\|\nonumber\\
&\leq 2\alpha\left\|\bM_0\right\|+2\alpha\beta^2\braket{\bxi(\bphi)|\bxi(\bphi)}\max_{m}\left\|\bM_m\right\|\nonumber\\
&\leq 2\alpha\left\|\bM_0\right\|+2\alpha\beta^2\max_{m}\left\|\bM_m\right\|.
\end{align}

The summands in the third term of~\eqref{eq:hessian:theta} can be upper bounded as:
\begin{align*}
&\left|\frac{\partial^2\mcL(\bz)}{\partial\theta_j\partial \phi_q}\right| 
\leq \alpha^2\beta^2\sum_{m=1}^M  \left|\frac{\partial |\xi_m(\bphi)|^2}{\partial \phi_q}\right| \left|\frac{\partial F_m(\btheta)}{\partial \theta_j}\right|\\
&\quad\leq \frac{\alpha^2\beta^2}{2}\sum_{m=1}^M \left||\xi_m(\bphi^+_q)|^2-|\xi_m(\bphi^-_q)|^2\right|\cdot \left\|\bM_m\right\|\\
&\quad\leq \frac{\alpha^2\beta^2}{2}\left(\sum_{m=1}^M  |\xi_m(\bphi^+_q)|^2+\sum_{m=1}^M |\xi_m(\bphi^-_q)|^2\right)\left\|\bM_m\right\|\\
&\quad\leq \alpha^2\beta^2\max_{m}\left\|\bM_m\right\|
\end{align*}
for all $q$. Then, the third term of~\eqref{eq:hessian:theta} is bounded as
\begin{equation}\label{eq:theta:phi}
\sum_{q=1}^Q\left|\frac{\partial^2\mcL(\bz)}{\partial\theta_j\partial \phi_q}\right| \leq Q\alpha^2\beta^2\max_{m}\left\|\bM_m\right\|.
\end{equation}

The fourth term in~\eqref{eq:hessian:theta} can be bounded as
\begin{align}\label{eq:theta:beta}
\left|\frac{\partial^2\mcL(\bz)}{\partial\theta_j\partial \beta}\right|&\leq 2\alpha^2\beta \sum_{m=1}^M|\xi_m(\bphi)|^2\left|\frac{\partial F_m(\btheta)}{\partial\theta_j}\right|\nonumber\\
&\leq 2\alpha^2\beta\max_{m}\left\|\bM_m\right\|.
\end{align}

Combining~\eqref{eq:theta:theta}--\eqref{eq:theta:beta} yields
\begin{align}\label{eq:bound:theta}
&\max_{j\in\{1,\ldots,P\}}\left\|\be_j^\top \bJ(\bz)\right\|_1 \leq(P\alpha^2+2\alpha)\left\|\bM_0\right\|\nonumber\\
&+\left(P\alpha^2\beta^2+Q\alpha^2\beta^2+2\alpha\beta^2+2\alpha^2\beta\right)\max_{m}\left\|\bM_m\right\|.
\end{align}

For the remaining rows of $\bJ(\bz)$ corresponding to the partial Hessian matrices $\nabla_{\alpha,\bz}^2\mcL(\bz)$, $-\nabla_{\bphi,\bz}^2\mcL(\bz)$, and $-\nabla_{\beta,\bz}^2\mcL(\bz)$, the norms $\left\|\be_j^\top\bJ(\bz)\right\|_1$ can be upper bounded using similar arguments. A key point here is that $\left\|\bM\right\|=\max_m\left\|\bM_m\right\|$ because $\bM$ in \eqref{eq:M} is block diagonal. To avoid repetition and tedious algebra, we omit the detailed derivations and present the final bounds.
\begin{align}\label{eq:bound:alpha}
&\|\be_{P+1}^\top\bJ(\bz)\|_1\leq (2P\alpha+2)\left\|\bM_0\right\|\nonumber\\
&+(2P\alpha\beta^2+4Q\alpha \beta^2+2\beta^2+4\alpha\beta)\max_{m}\left\|\bM_m\right\|,
\end{align}
\begin{align}\label{eq:bound:phi}
&\max_{j\in\{P+2:P+Q+1\}}\left\|\be_j^\top\bJ(\bz)\right\|_1 \leq (Q\beta^2+2\beta) \max_m |b_m|\nonumber\\
&\quad\quad+\left(P\alpha^2\beta^2+Q\alpha^2\beta^2+2\alpha\beta^2+2\alpha^2\beta\right)\max_{m}\left\|\bM_m\right\|,
\end{align}
\begin{align}\label{eq:bound:beta}
&\|\be_{P+Q+2}^\top\bJ(\bz)\|_1 \leq (2Q\beta+2)\max_m|b_m|\nonumber\\
&\quad+(P\alpha^2\beta+2Q\alpha^2\beta+4\alpha\beta+2\alpha^2)\max_{m}\left\|\bM_m\right\|.
\end{align}

According to~\eqref {eq:lipschitz:bound}, the Lipschitz constant $L$ of $\bg(\bz)$ is the maximum of the right-hand sides of~\eqref{eq:bound:theta}--\eqref{eq:bound:beta}. If $\alpha \geq 4$, the right-hand side (RHS) of \eqref{eq:bound:theta} is larger than the RHS of \eqref{eq:bound:alpha}. The condition $\alpha \geq 4$ holds trivially if the problem size is $N>32$. Moreover, if $\beta \geq 2$, the RHS of \eqref{eq:bound:phi} is larger than that of \eqref{eq:bound:beta}. The condition $\beta \geq 2$ is expected to occur for problems with many constraints. Therefore, we are left with the RHS of \eqref{eq:bound:theta} and \eqref{eq:bound:phi}, whose maximum can be compactly expressed as
\begin{align*}\label{eq:alpha&phi}                  
&\left(P\alpha^2\beta^2+Q\alpha^2\beta^2+2\alpha\beta^2+2\alpha^2\beta\right)\max_{m}\left\|\bM_m\right\|\nonumber\\
&+\max\{(P\alpha^2+2\alpha)\left\|\bM_0\right\|,(Q\beta^2+2\beta) \max_m |b_m|\}.
\end{align*}
Then, the claim of Lemma~\ref{le:lipschitz} follows readily. 
\begin{table*}[t]
    \centering
    \begin{tabular}{|c|l|c|c|c|c|c|c|}
        \hline
        No. & architecture of 1 layer & $L$ (primal) & $P$ & error (primal) & $L$ (dual) & $Q$ & error (dual) \\
        \hline
        1 & $\bR_x-$CX & 20 & 120 & $9\times 10^{-1}$ & 35 & 315 &  $8\times 10^{-5}$\\
        2 & $\bR_y-$CX & 20 & 120 & $5\times 10^{-3}$ & 35 & 315 & $2\times 10^{-5}$  \\
        3 & $\bR_z-$CX & 20 & 120 & 1 & 35 & 315 &  $2\times 10^{-3}$\\
        4 & $\bR_x-$CX$-\bR_y-$CX & 10 & 120 & $3\times 10^{-4}$ & 18 & 324 &  $7\times 10^{-5}$\\
        5 & $\bR_x-$CX$-\bR_z-$CX & 10 & 120 & $3\times 10^{-4}$ & 18 & 324 &  $7\times 10^{-5}$\\
        6 & $\bR_y-$CX$-\bR_z-$CX & 10 & 120 & $2\times 10^{-4}$ & 18 & 324 & $7\times 10^{-5}$ \\
        7 & $\bR_x-$CX$-\bR_y-$CX$-\bR_z-$CX & 6 & 108 & $3\times 10^{-3}$ & 12 & 324 & $7\times 10^{-5}$ \\
        8 & $\bR_x-\bR_y-\bR_z-$CX & 7 & 126 & $10^{-3}$ & 12 & 324 & $4\times 10^{-6}$ \\
        \hline
    \end{tabular}
    \caption{Tested PQC architectures.}
    \label{tab:architectures}
\end{table*}
\subsection{Proof of Lemma~\ref{le:var}}\label{sec:app:th}
We first upper bound the variance of the sample-average estimators of the three terms in~\eqref{eq:Lagrangian1}. Note that the three terms in~\eqref{eq:Lagrangian1} can be decomposed into quantum expectations across $c$. This is obvious for $F_0(\btheta)$ and $G(\bphi)$. Regarding $F(\btheta,\bphi)$, we obtain from \eqref{eq:term2:measure}:
\begin{align*}
&F(\btheta,\bphi)=\sum_{c=1}^{2C-1}\sum_{m=1}^M |\xi_m(\bphi)|^2\braket{\bpsi_c(\btheta)|\bLambda_m^c|\bpsi_c(\btheta)}\\
&=\sum_{c=1}^{2C-1} \braket{\bxi(\bphi), \bpsi_c(\btheta)|\left(\sum_{m=1}^M \be_m \be_m^\top \otimes \bLambda^c_m\right)|\bxi(\bphi),\bpsi_c(\btheta)}.
\end{align*}

Because the three expectations in~\eqref{eq:Lagrangian1} are defined over different states and observables, we consider a general expectation $E\coloneqq \braket{\bzeta|\bN|\bzeta}$, where $\ket{\bzeta}$ and $\bN$ are dimensionally compatible and $\bN$ is a Hermitian matrix with the same color decomposition as matrices $\bM_m$'s. The expectation $E$ can be color-decomposed as 
\begin{equation}\label{eq:E}
E=\sum_{c=1}^{2C-1}\braket{\bzeta|\bN^c|\bzeta}=\sum_{c=1}^{2C-1}\braket{\bzeta_c|\bLambda^c|\bzeta_c},
\end{equation}
where $\bN^c=\bU_c \bLambda^c \bU_c^{\dag}$ is the eigendecomposition of $\bN^c$ and $\ket{\bzeta_c}=\bU_c^{\dag} \ket{\bzeta}$.

Upon measuring $\ket{\bzeta_c}$ in the computational basis using $S$ measurement samples, the binary outcome $\ket{i}$ is observed with probability $|\zeta^i_c|^2\coloneqq|\braket{i|\bzeta_c}|^2$. For each sample $s$, define a random variable $\htE_s^c$, taking the value $\braket{i|\bLambda^c|i}$ when outcome $\ket{i}$ is observed while sampling $\ket{\bzeta_c}$. Subsequently, the sample-average estimator of $E$ can be expressed as
\begin{equation}\label{eq:Ehat}
    \htE=\frac{1}{S}\sum_{c=1}^{2C-1}\sum_{s=1}^S \htE_s^c.
\end{equation}
The variance of $\htE$ is bounded in the next lemma. 

\begin{lemma}\label{le:mean&var}
The variance of the estimator in~\eqref{eq:Ehat} can be computed as:
\begin{align}\label{eq:F0_var}
\Var(\htE)&=\frac{1}{S}\sum_{c=1}^{2C-1}(\braket{\bzeta_c|(\bLambda^c)^2|\bzeta_c}-\braket{\bzeta_c|\bLambda^c|\bzeta_c}^2),
\end{align}
and can be upper bounded as
\begin{equation} \label{eq:varE:bound}
    \Var(\htE) \leq \frac{1}{S}\sum_{c=1}^{2C-1}\left\|\bN^c\right\|^2.
\end{equation}
\end{lemma}

\begin{proof}
The variance of each $\htE_s^c$ can be computed and upper bounded as follows:
\begin{align*}
&\Var(\htE_s^c)
=\mathbb{E}[(\htE_s^c)^2]-(\mathbb{E}[\htE_s^c])^2\\
&=\sum_{i=0}^{N-1} |\zeta^i_c|^2 \braket{i|\bLambda^c|i}^2-\left(\sum_{i=0}^{N-1}|\zeta^i_c|^2 \braket{i|\bLambda^c|i}\right)^2\\
&=\braket{\bzeta_c|(\bLambda^c)^2|\bzeta_c}-\braket{\bzeta_c|\bLambda^c|\bzeta_c}^2\\
&\leq \braket{\bzeta_c|(\bLambda^c)^2|\bzeta_c}\\
&\leq \|\bN^c\|^2.
\end{align*}
Because the random variables $\htE_s^c$ are independent across $c$ and $s$, we get that
\begin{align*}
\Var(\htE)=\frac{1}{S^2}\sum_{c=1}^{2C-1}\sum_{s=1}^S \Var(\htE_s^c)=\frac{1}{S}\sum_{c=1}^{2C-1}\Var(\htE_s^c)
\end{align*}
for any $s$. Substituting $\Var(\htE_s^c)$ above yields the expression in~\eqref{eq:F0_var}. The final bound in~\eqref{eq:varE:bound} follows readily.
\end{proof}

We are now ready to upper bound the variance of the estimate $\hbg(\bz)$ per Theorem~\ref{th:minimax}.

\begin{proof}[Proof of Lemma~\ref{le:var}]
The RHS of~\eqref{eq:sigma} can be expanded as
\begin{align}
&\mathbb{E}\left[\left\|\hbg(\bz)-\bg(\bz)\right\|_2^2\right]=\nonumber \\
&\underbrace{\mathbb{E}[|\hat{\nabla}_{\alpha}\mcL(\bz)-\nabla_{\alpha}\mcL(\bz)|^2]}_{\coloneqq \Var(\hat{\nabla}_{\alpha}\mcL(\bz))}+\mathbb{E}[\|\hat{\nabla}_{\btheta}\mcL(\bz)-\nabla_{\btheta}\mcL(\bz)\|_2^2]\nonumber \\
&+\mathbb{E}[\|\hat{\nabla}_{\bphi}\mcL(\bz)-\nabla_{\bphi}\mcL(\bz)\|_2^2]+ \underbrace{\mathbb{E}[|\hat{\nabla}_{\beta}\mcL(\bz)-\nabla_{\beta}\mcL(\bz)|^2]}_{\coloneqq \Var(\hat{\nabla}_{\beta}\mcL(\bz))}.\label{eq:grad:z}
\end{align}
Among four terms in~\eqref{eq:grad:z}, the terms $\Var(\hat{\nabla}_{\alpha}\mcL(\bz))$ and $\Var(\hat{\nabla}_{\beta}\mcL(\bz))$ are simpler to bound. Let $\htF_0(\btheta)$ and $\htF(\btheta,\bphi)$ denote the sample-based estimates of $F_0(\btheta)$ and $F(\btheta,\bphi)$, respectively. Then, we can express the first term in~\eqref{eq:grad:z} as:
\begin{align}\label{eq:grad:alpha1}
&\Var(\hat{\nabla}_{\alpha}\mcL(\bz))=\Var\left(2\alpha\htF_0(\btheta)+2\alpha\beta^2\htF(\btheta,\bphi)\right)\nonumber\\
&=4\alpha^2\Var(\htF_0(\btheta))+4\alpha^2\beta^4\Var(\htF(\btheta,\bphi)).
\end{align}
It follows from~\eqref{eq:varE:bound} that
\begin{equation} \label{eq:var1:bound}   
\Var(\htF_0(\btheta))\leq \frac{1}{S}\sum_{c=1}^{2C-1}\|\bM^c_0\|^2.
\end{equation}

Regarding $\Var(\htF(\btheta,\bphi))$, because $\bM$ is a block-diagonal matrix with blocks $\bM_m$, it follows that
\begin{align}\label{eq:var2:bound}
    &\Var(\htF(\btheta,\bphi))\leq \frac{1}{S}\sum_{c=1}^{2C-1}\max_{m}\|\bM^c_m\|^2.
\end{align}

Substituting \eqref{eq:var1:bound}--\eqref{eq:var2:bound} into \eqref{eq:grad:alpha1} provides
\begin{align}
    \Var(\hat{\nabla}_{\alpha}\mcL(\bz))&\leq\frac{4\alpha^2}{S}\sum_{c=1}^{2C-1}\|\bM^c_0\|^2\nonumber\\
    &+\frac{4\alpha^2\beta^4}{S}\sum_{c=1}^{2C-1}\max_{m}\|\bM^c_m\|^2\label{eq:grad:alpha}.
\end{align}

Let $\hat{\partial}_{\theta_p}F_0(\btheta)$ and $\hat{\partial}_{\theta_p}F(\btheta,\bphi)$ denote the sample-based estimates of $\partial_{\theta_p}F_0(\btheta)$ and ${\partial}_{\theta_p}F(\btheta,\bphi)$, respectively. The second expectation in~\eqref{eq:grad:z} entails computing 
\begin{align}
   &\mathbb{E}[\|\hat{\nabla}_{\btheta}\mcL(\bz)-\nabla_{\btheta}\mcL(\bz)\|_2^2]=\sum_{p=1}^P\Var(\hat{\partial}_{\theta_p}\mcL(\bz))\nonumber \\
   &=\sum_{p=1}^P\Var\left(\alpha^2\hat{\partial}_{\theta_p}F_0(\btheta)+\alpha^2\beta^2 \hat{\partial}_{\theta_p}F(\btheta,\bphi)\right)\nonumber \\
    &=\alpha^4\sum_{p=1}^P \left(\Var(\hat{\partial}_{\theta_p}F_0(\btheta))+\beta^4 \Var(\hat{\partial}_{\theta_p} F(\btheta,\bphi))\right)\nonumber \\
    &=\frac{\alpha^4}{4}\sum_{p=1}^P \left(\Var(\htF_0(\btheta^+_p))+\Var (\htF_0(\btheta^-_p))\right)\nonumber\\
    &+\frac{\alpha^4\beta^4}{4}\sum_{p=1}^P \left(\Var(\htF(\btheta_p^+,\bphi))+\Var(\htF(\btheta_p^-,\bphi)\right)\nonumber\\
    &\leq \frac{P\alpha^4}{2S}\sum_{c=1}^{2C-1}\|\bM^c_0\|^2+\frac{P\alpha^4\beta^4}{2S}\sum_{c=1}^{2C-1}\max_{m}\|\bM^c_m\|^2.
\label{eq:grad:theta}
\end{align}

Analogously, the third and fourth expectations in~\eqref{eq:grad:z} can be upper bounded as
\begin{align}\label{eq:grad:phi}
    &\mathbb{E}[\|\hat{\nabla}_{\bphi}\mcL(\bz)-\nabla_{\bphi}\mcL(\bz)\|_2^2]\leq\nonumber\\
    &\frac{Q\beta^4}{2S}\max_{m}b_m^2+\frac{Q\alpha^4\beta^4}{2S}\sum_{c=1}^{2C-1}\max_{m}\|\bM^c_m\|^2,
\end{align}

\begin{align}   
&\Var(\hat{\nabla}_{\beta}\mcL(\bz))\leq \nonumber\\
&\frac{4\beta^2}{S}\max_{m}b_m^2+\frac{4\alpha^4\beta^2}{S}\sum_{c=1}^{2C-1}\max_{m}\|\bM^c_m\|^2.\label{eq:grad:beta}
\end{align}

Summing \eqref{eq:grad:alpha}--\eqref{eq:grad:beta}, the RHS of~\eqref{eq:sigma} is bounded as:
\begin{align*} 
&\mathbb{E}\left[\left\|\hbg(\bz)-\bg(\bz)\right\|_2^2\right]\leq \nonumber\\
&\frac{Q\bar{\beta}^4+8\bar{\beta}^2}{2S}\max_{m}b_m^2+\frac{8\bar{\alpha}^2+P\bar{\alpha}^4}{2S}\sum_{c=1}^{2C-1}\|\bM^c_0\|^2+\nonumber\\
&\frac{8\bar{\alpha}^2 \bar{\beta}^4+8\bar{\alpha}^4\bar{\beta}^2+(P+Q)\bar{\alpha}^4 \bar{\beta}^4}{2S}\sum_{c=1}^{2C-1}\max_{m}\|\bM^c_m\|^2.
\end{align*}
\end{proof}

\subsection{PQC architectures}\label{sec:app:vqc}
The tested PQC architectures are reported in Table~\ref{tab:architectures}.

\vfill 
\end{document}